\documentclass[letter,11pt,final]{article}
\usepackage{typearea}
\usepackage{pifont}
\newcommand{\cmark}{\ding{51}}
\newcommand{\xmark}{\ding{55}}
\usepackage[tableposition=bottom]{caption}
\usepackage{booktabs}

\typearea{14}

\usepackage[numbers,sort,compress]{natbib}
\usepackage[T1]{fontenc}
\usepackage{amsmath,amsthm,amssymb,thm-restate}

\usepackage{authblk}
\usepackage{algorithm}
\usepackage{algorithmicx}
\usepackage[noend]{algpseudocode}
\usepackage[colorlinks=true,citecolor=blue,urlcolor=blue,bookmarks=false]{hyperref}

\usepackage{bigstrut}
\usepackage{multirow}
\usepackage{subcaption}

\usepackage{dsfont}
\usepackage{tikz}
\usetikzlibrary{calc}

\usepackage{enumitem}
\newcommand{\subscript}[2]{$#1 _ #2$}
\usepackage{bm}

\usepackage[framemethod=tikz]{mdframed}

\usepackage{cleveref}
\usepackage{xspace}
\usepackage{xcolor}
\usepackage[textsize=tiny]{todonotes}

\theoremstyle{plain}

\newtheorem{thm}{Theorem}[section]
\newtheorem{cor}[thm]{Corollary}
\newtheorem{prop}[thm]{Proposition}

\newtheorem{lem}[thm]{Lemma}
\newtheorem{Def}[thm]{Definition}
\newtheorem{obs}[thm]{Observation}

\newtheorem{remark}[thm]{Remark}

\newcommand{\poly}{\mathrm{poly}\xspace}

\newcommand{\eat}[1]{}
\newcommand{\eps}{\epsilon}

\newcommand{\calA}{\mathcal{A}}

\newenvironment{wrapper}[1]
{
	\smallskip
	\begin{center}
		\begin{minipage}{\linewidth}
			\begin{mdframed}[hidealllines=true, backgroundcolor=gray!20, leftmargin=0cm,innerleftmargin=0.375cm,innerrightmargin=0.375cm,innertopmargin=0.375cm,innerbottommargin=0.375cm,roundcorner=10pt]
				#1}
			{\end{mdframed}
		\end{minipage}
	\end{center}
	\smallskip
}

\newcommand{\inserttwod}{\textsc{insert}}
\newcommand{\remove}{\textsc{remove}}

\newcommand{\leqstar}{\stackrel{\star}{\leq}}
\newcommand{\geqstar}{\stackrel{\star}{\geq}}

\title{Beating the Folklore Algorithm for Dynamic Matching}

\author{Mohammad Roghani}
\author{Amin Saberi}
\author{David Wajc}
\affil{Stanford University\\ \{roghani,saberi,wajc\}@stanford.edu}
\author{}
\date{}

\begin{document}

\maketitle

\pagenumbering{gobble}
\begin{abstract}
    The maximum matching problem in dynamic graphs subject to edge updates (insertions and deletions) has received much attention over the last few years; a multitude of approximation/time tradeoffs were obtained, improving upon the folklore algorithm, which maintains a maximal (and hence $2$-approximate) matching in $O(n)$ worst-case update time in $n$-node graphs. 
    
    \medskip 
    
    We present the first deterministic algorithm which outperforms the folklore algorithm in terms of {\em both} approximation ratio and worst-case update time. Specifically, we give a $(2-\Omega(1))$-approximate algorithm with $O(m^{3/8})=O(n^{3/4})$ worst-case update time in $n$-node, $m$-edge graphs. For sufficiently small constant $\eps>0$, no deterministic $(2+\eps)$-approximate algorithm with worst-case update time $O(n^{0.99})$ was known. Our second result is the first deterministic $(2+\eps)$-approximate \emph{weighted} matching algorithm with $O_\eps(1)\cdot O(\sqrt[4]{m}) = O_\eps(1)\cdot O(\sqrt{n})$ worst-case update time. Neither of our results were previously known to be achievable by a randomized algorithm against an adaptive adversary.
    
    \medskip
    
    Our main technical contributions are threefold: 
    first, we characterize the tight cases for \emph{kernels}, which are the well-studied matching sparsifiers underlying much of the $(2+\eps)$-approximate dynamic matching literature. This characterization, together with multiple ideas---old and new---underlies our result for breaking the approximation barrier of $2$.
    Our second technical contribution is the first example of a dynamic matching algorithm whose running time is improved due to improving the \emph{recourse} of other dynamic matching algorithms.
    Finally, we show how to use dynamic bipartite matching algorithms as black-box subroutines for dynamic matching in general graphs without incurring the natural $\frac{3}{2}$ factor in the approximation ratio which such approaches naturally incur (reminiscent of the integrality gap of the fractional matching polytope in general graphs).
    \color{black}
\end{abstract}

\newpage 
\pagenumbering{arabic}
\section{Introduction}

We study the dynamic (weighted) matching problem, where our goal is to maintain an approximately maximum (weight) matching subject to edge updates (insertions and deletions) in a dynamically evolving graph.

\smallskip

For approximation and update time, there are two natural barriers for the dynamic matching problem: an approximation ratio of $2$, and an update time of $O(n)$. Both bounds are achieved by a folklore deterministic algorithm maintaining a maximal matching in worst-case update time $O(n)$ (see \Cref{fig:folklore}).
Moreover, both bounds stand at the frontier of known time/approximation tradeoffs for randomized and deterministic dynamic matching algorithms.
For example, an approximation ratio of $2$ is the best approximation ratio known to be achievable in polylogarithmic worst-case update time \cite{bernstein2019deamortization} or constant amortized time \cite{solomon2016fully} for randomized algorithms.\footnote{A dynamic algorithm has amortized update time $f(n)$ if for any sequence of $t$ updates starting with the empty graph, the algorithm takes $t\cdot f(n)$ time.}
Similarly, an approximation ratio of $2+\eps$ is the best approximation ratio known to be achievable in polylogarithmic amortized update time deterministically \cite{bhattacharya2021deterministic,bhattacharya2016new} or worst-case polylogarithmic update time without the oblivious adversary assumption \cite{wajc2020rounding}.\footnote{The oblivious adversary assumption stipulates that the update sequence is generated non-adaptively, and in particular each update is independent of the algorithm's previous random coin tosses. This assumption, needed for the analysis of many randomized dynamic algorithms, rules out their black-box use in many applications. See \cite{nanongkai2017dynamic,wajc2020rounding} for discussions.}
On the other hand, $O(\sqrt{m})=O(n)$ is the current best worst-case update time for deterministic algorithms achieving better than $2$ approximation \cite{gupta2013fully,peleg2016dynamic,neiman2016simple}.
Moreover, an update time of $\Omega(n)$ (ignoring sub-polynomial factors) is known to be necessary for (exact) dynamic matching, assuming any one of several widely believed conjectures, including the Strong Exponential Time Hypothesis, and the 3SUM, APSP, and OMv conjectures \cite{dahlgaard2016hardness,henzinger2015unifying,kopelowitz2016higher}.

\begin{figure}[H]
\begin{center}
\fbox{\begin{minipage}{0.95\linewidth}
    \textsc{Init:} $M\gets \emptyset$.
    
    \textsc{Insert}$(e)$: If both $u,v\in e$ are unmatched in $M$, then $M\gets M\cup\{e\}$.
    
    \textsc{Delete}$(e)$: If $e\in M$, for each $v\in e$, if $v$ has an unmatched neighbor $w$, then $M\gets M\cup\{(v,w)\}$.
    \end{minipage}}
\end{center}
\vspace{-0.5cm}
\caption{The Folklore Algorithm}\label{fig:folklore}
\end{figure}

Given the above, breaking the approximation ratio of $2$ and worst-case update time of $O(n)$ barriers stand as two recurring goals of the rich literature on the dynamic matching problem.
A concentrated effort, starting with the influential work of \citet{onak2010maintaining}, has resulted in numerous algorithms breaking either one of these barriers individually \cite{baswana2015fully,onak2010maintaining,solomon2016fully,charikar2018fully,arar2018dynamic,bernstein2019deamortization,bhattacharya2016new,bhattacharya2018deterministic,bhattacharya2021deterministic,gupta2013fully,neiman2016simple,bernstein2015fully,bernstein2016faster,peleg2016dynamic,wajc2020rounding,behnezhad2019fully,bhattacharya2020deterministic}. However, despite this long line of work, both barriers were not known to be surpassable simultaneously without assuming sparsity or bipartiteness \cite{gupta2013fully,bernstein2015fully,neiman2016simple,peleg2016dynamic,wajc2020rounding}, settling for amortized update time \cite{bernstein2016faster}, or using randomness and the oblivious adversary assumption \cite{behnezhad2020fully}. Whether or not there exists a deterministic algorithm that beats the trivial folklore $O(n)$-time maximal matching algorithm both in terms of worst-case update time and approximation ratio, thus simultaneously breaking these two natural barriers for this problem in its full generality, remained a vexing open problem.

\subsection{Results}

We resolve the above open problem, and present the first deterministic algorithm which outperforms the folklore dynamic matching algorithm,  in terms of both approximation and worst-case update time. Our main result is the following. 

\begin{wrapper}
    \begin{restatable}{thm}{beatfoklore}(Beating the Folkore Algorithm)\label{thm:det-beat-2}
        There exists a deterministic $1.999999$-approximate dynamic matching algorithm with worst-case update time $O(m^{3/8})$.
    \end{restatable}
\end{wrapper}

Prior to this work, no deterministic algorithm was known to   achieve $2+\eps$ approximation in $O_\eps(n^{0.99})$ worst-case update time, where $O_\eps(\cdot)$ suppresses dependencies on $\eps$. As a byproduct of the algorithm given for \Cref{thm:det-beat-2}, 
we obtain a secondary result, rectifying this state of affairs for both the matching and weighted matching problem and present the first $(2+\eps)$-approximate matching algorithm with (polynomially) sub-linear update time for any constant $\eps>0$.

\begin{restatable}{thm}{twoplusepsalgo}\label{thm:fast-kernel-algo}(Informal for Weighted)
    There exist deterministic $(2+\eps)$-approximate dynamic matching and weighted matching algorithms with worst-case update time $O_\eps(1)\cdot O(\sqrt[4]{m}) = O_\eps(1) \cdot O(\sqrt{n})$.
\end{restatable}

For constant $\eps>0$, the update time of the algorithms of \Cref{thm:fast-kernel-algo} is $O(\sqrt[4]{m})$. Such worst-case update time was only previously known to yield a $(9/4+\eps)$ approximation (see \cite[Section 7]{bernstein2021framework} and \Cref{sec:bipartite:reduction}), or worse \cite{wajc2020rounding}.

\smallskip

We contrast our results with prior linear- and sublinear-time deterministic algorithms in \Cref{table:Results}.

\vspace{-0.3cm}
\begin{center}	
	\begin{table}[H]
		{
		\small
		\begin{center}
			\centering
			
			\begin{tabular}{ | c | c |  c |  c | c |}
				\hline				
				Approx. & Update Time & Worst Case & Notes  & Reference \bigstrut\\  
				\specialrule{.125em}{.0625em}{.0625em} 	
				
				$4+\epsilon$ & $O(\sqrt[3]{m}/\epsilon^2)$ & \cmark &  
				
				& \multirow{2}{*}{Bhattacharya et al.~(SODA '15) \cite{bhattacharya2018deterministic}\bigstrut} \\ 
				
				$3+\epsilon$ &
				$O(\sqrt{n}/\epsilon)$ & \xmark &  
				& \\ 	
				\hline	
				
				$9/4+\epsilon$ & $O(\sqrt[4]{m})\cdot \poly(1/\eps)$ \bigstrut  & \cmark & & Bernstein et al.~(STOC '21) \cite{bernstein2021framework} \bigstrut \\ \hline	
				
				$(2+\epsilon)\cdot c$ & $\tilde{O}(n^{1/c})\cdot \poly(1/\eps)$ \bigstrut  & \cmark &  $\forall c\geq 1$ & Wajc (STOC '20) \cite{wajc2020rounding} \bigstrut \\ \hline		
				
				\specialrule{.125em}{.0625em}{.0625em}

			   $2 + \eps$ & $O(\sqrt[4]{m}/\sqrt{\eps})$ & \cmark &   & \textbf{This Work} \bigstrut \\ 	\hline
				
				\multirow{2}{*}{$2+\eps$} & \multirow{2}{*}{$\poly(\log n,1/\eps)$} & \multirow{2}{*}{\xmark} & & Bhattacharya et al.~(STOC '16) \cite{bhattacharya2016new} \bigstrut \\
                 & & &  & Bhattacharya and Kiss (ICALP'21) \cite{bhattacharya2021deterministic} \bigstrut \\ 	\hline

			   $2$ & $O(n)$  & \cmark &  & Folklore \bigstrut \\ 	\hline					
				\specialrule{.125em}{.0625em}{.0625em}			
			   
			   $1.999999$ & $O(\sqrt{n}\cdot \sqrt[8]{m})$
			   & \cmark &   & \textbf{This Work} \bigstrut \\ 	\hline		

				$3/2 +\epsilon$ & $O(\sqrt[4]{m}/ \epsilon^{2.5})$
				& \xmark &  
				& Bernstein and Stein (SODA '16) \cite{bernstein2016faster} \bigstrut \\ 		\hline
				
				$3/2$ & $O(\sqrt{m})$
				& \cmark &   & Neiman and Solomon (STOC '13) \cite{neiman2016simple} \bigstrut \\ 		\hline
				
				$1+\epsilon$ & $O(\sqrt{m}/\epsilon^{2})$
				& \cmark &  & Gupta and Peng (FOCS '13) \cite{gupta2013fully}\bigstrut \\ 		\hline
			\end{tabular}
		\end{center}}	
		
		\captionsetup{justification=centering}
		\caption{Known $O(n)$ Time Deterministic Dynamic Matching Algorithms in General Graphs
		\\(References are to the latest publication, with the first publication venue in parentheses)}	
		\label{table:Results}	
	\end{table}
\end{center}

\subsection{Our Approach in a Nutshell}

\textbf{Warm-up: Faster $(2+\eps)$-Approximate Matching.}
Our starting point is the highly successful matching sparsifiers of \citet*{bhattacharya2018deterministic}---termed \emph{kernels}---previously used in numerous dynamic matching algorithms \cite{bhattacharya2018deterministic,arar2018dynamic,wajc2020rounding,bhattacharya2016new,bernstein2021framework}.
Kernels are $(2+\eps)$-approximate matching sparsifiers (i.e., these subgraphs contain a $(2+\eps)$-approximate matching),
of small maximum degree, $d=o(n)$.
To obtain approximation ratios close to $(2+\eps)$, it is therefore natural to combine kernel-maintenance algorithms with known near-maximum (i.e., $(1+\eps)$-approximate) algorithms with worst-case update time linear in the maximum degree \cite{gupta2013fully,peleg2016dynamic}.
This seems to suggest an $O(T+d)$ time $(2+O(\eps))$-approximate algorithm, where $T$ is the update time for maintaining the kernel.

\smallskip 

Unfortunately, combining these two ideas does not immediately result in an $O(T+d)$ update time. This is because  near-maximum matching algorithms take  $O(d)$ \emph{time per update to the kernel}, and not per update to the input dynamic graph $G$. Consequently, if $c$ is the number of changes to the kernel per update to $G$, then combining these algorithms yields an update time of $O(T+c\cdot d)$. 
Prior deterministic kernel maintenance algorithms \cite{bhattacharya2018deterministic} all had $c\cdot d = \Omega(n)$ in the worst case. Consequently, no deterministic $(2+\eps)$-approximate algorithm with worst-case $o(n)$ time was previously known. 

\smallskip

In \Cref{sec:kernel-maintenance}, we show how to de-amortize an $O_\eps(n/d)$ update time kernel maintenance algorithm of 
\citet{bhattacharya2018deterministic} (and speed it up by an $O(n/\sqrt{m})$ factor in \Cref{sec:kernel-in-sqrtm-time}), while making only $c=O(1)$ updates to the kernel $K$ per update to $G$. Setting the kernel's (tunable) maximum degree to $d=\sqrt[4]{m}$ yields our sublinear $(2+\eps)$-approximate unweighted algorithm of \Cref{thm:fast-kernel-algo}.
Our weighted algorithm of the same theorem (see \Cref{sec:application:MWM}) uses our kernel algorithm and the recent dynamic weighted matching framework of \citet{bernstein2021framework}.

\medskip 
\noindent\textbf{Breaking the Barrier of Two.}
If the kernel $K$ which we maintain \emph{happens to be} a $(2-\eps)$-approximate matching sparsifier, then a near-maximum matching in $K$ would be a $(2-\Omega(\eps))$-approximate matching in $G$, maintainable in worst-case $o(n)$ update time. 
Unfortunately, the $(2+\eps)$ approximation ratio can be tight for kernels. 
Nonetheless, we show that it is possible to use kernels to obtain better-than-two-approximate matchings dynamically.

\smallskip 

Our high-level approach is a natural one: we find a bounded-degree subgraph $A$, whose union with the kernel contains, say, $2\eps\cdot \mu(G)$ disjoint augmenting paths with respect to some maximum matching in the kernel $K$. Since the kernel $K$ is $(2+\eps)$-approximate, 
the augmented kernel $AK:=A\cup K$ is therefore $(2-\eps)$-approximate, whether or not the kernel is $(2-\eps)$-approximate. Therefore,  we can maintain in $o(n)$ worst-case update time, a $(1+O(\eps))$-approximate matching in this low-degree subgraph, $AK$, giving a $(2-O(\eps))$-approximate matching in $G$.

\smallskip

The above approach seems simple, but its dynamic implementation presents two challenges.
The first challenge is to identify a \emph{dynamically maintainable} sparse subgraph $A$ whose edges ``augment'' the kernel. 
The second challenge is to guarantee that following each update to $G$, the worst-case number of changes to the augmented kernel $AK=A\cup K$ is small. This is crucial, since the time to update the output matching, i.e., a near-maximum matching in $AK$, is proportional to the product of the maximum degree of $AK$ and the number of edge updates to $AK$ \emph{per update to $G$.}

\medskip\noindent\textbf{Identifying Good Augmentations.}
Our first key observation is a structural characterization of kernels which are worse than $(2-\eps)$-approximate.
First, we show that such a kernel $K$ contains a maximum matching $M$ such that nearly all connected components of $M\cup M^*$ (where $M^*$ is any maximum matching in $G$) are augmenting paths of length three w.r.t.~$M$. 
Moreover, in the vast majority of these augmenting paths, the edges of $M^*$ connect a high- and a low-degree node in $K$.
This motivates us to compute a large matching $M'$ in the \emph{bipartite} subgraph between high- and low-degree nodes in $K$, in the hope that such an $M'$ contains edges of $G\setminus K$ whose addition to the kernel $K$ increases its maximum matching size.
(Foreshadowing our use of the dynamic bipartite matching algorithm of \citet{bernstein2015fully}, we focus on $(3/2+\eps)$-approximate matchings in this subgraph.) 
While a single such matching $M'$ is insufficient, we show in \Cref{sec:two-matchings-suffice} that the union of two $(3/2+\eps)$-approximate matchings in two similarly-defined bipartite subgraphs do constitute a good bounded-degree augmentation $A$ of the kernel. That is, the augmented kernel, $AK=A\cup K$, contains a $(2-\eps)$-approximate matching. 
This allows us to use kernels and $(3/2+\eps)$-approximate bipartite matching algorithms to obtain $(2-\eps)$-approximate matchings in general graphs in a \emph{static} setting. 

\medskip\noindent\textbf{Dynamizing Our Approach.}
An intricate part of this work is in dynamizing our outlined approach. For this, we must (i) make few changes to the bipartite subgraphs derived from $K$ (as each such change will cost the update time of the bipartite matching algorithm), and (ii) make few changes to the matchings computed in these subgraphs (as each such change will cost the update time of the bounded-degree algorithm in $K$). 
However, the bipartite subgraphs discussed in our static approach above are defined by degrees in the kernel $K$, which may change abruptly, resulting in \emph{vertex updates} (vertex insertions and deletions) in these bipartite subgraphs. This is problematic, as no $(2-\Omega(1))$-approximate algorithm is known with $o(n)$ vertex update time, even in bipartite graphs. 

\medskip\noindent\textbf{Approximate Degrees, and Star Updates.} To overcome the above bottleneck, we extend our static approach, and
prove that a sparse augmentation $A$ of the kernel $K$ is also obtained from the union of large matchings in some $O(1)$ bipartite subgraphs, $B_1,B_2,\dots,B_{O(1)}$, defined based on \emph{approximate degrees} in $K$ (see \Cref{sec:main-approx-degrees} for precise definition).
The key advantage of basing these subgraphs on approximate degrees is that (i) we can maintain such approximate degrees, and hence these subgraphs, in $o(n)$ update time (see \Cref{sec: approx-deg-maintenance}), and (ii) these bipartite subgraphs $\{B_i\}_i$ change less abruptly than their exact-degree-based counterparts. In particular, each change in $G$ (and hence in $K$) results in \emph{local} changes to the approximate degrees. Specifically, each update to $G$ causes the bipartite subgraphs $\{B_i\}_i$ to change by \emph{star updates}---addition/removal edges of a star graph with $b=o(n)$ edges.

\smallskip 
These star updates are doubly beneficial: first, they allow us to maintain, using the bipartite algorithm of \citet{bernstein2015fully}, a $(3/2+\eps)$-approximate matching $M_i$ in each $B_i$ in worst-case update time $b\cdot O_\eps(\sqrt[4]{m})$, which for our values of $b$ is sublinear in $n$. Second, since a star update only increases or decreases the size of any matching by at most one, we can appeal to (a slight extension of) the recent framework of \citet{solomon2021generalized} for decreasing \emph{recourse} (number of changes to the output per update) to guarantee that each matching $M_i$ only changes by $O(1/\eps)=O(1)$ edges per star update, and hence per update to $G$. 
We then take the union of these dynamic matchings $\{M_i\}_i$ to be our dynamic augmentation, $A=\bigcup_i M_i$, which changes by $O(1)$ edges per update to $G$. Since the kernel likewise changes by $O(1)$ edges per update to $G$, we have that the augmented kernel $AK=A\cup K$, which has maximum degree $d+O(1)$, has $c=O(1)$ edge changes per update to $G$. This then allows us to maintain a near-maximum matching in $AK$ in an additional $O(c\cdot d) = O(d) = o(n)$ worst-case time per update in $G$. By our structural results of \Cref{sec:kernels++}, this last matching is a $(2-\Omega(\eps))$-approximate matching in $G$, yielding our main result, \Cref{thm:det-beat-2}.

\subsubsection{Conceptual Contributions}
As the above overview suggests, our main result requires a careful combination of a number of ideas. Here we outline the key novel ideas. 

\smallskip 

Our first new contribution is a structural characterization of ``tight'' kernels, i.e., kernels whose approximation is no better than roughly two approximate. Given the use of kernels in the dynamic matching literature \cite{bhattacharya2018deterministic,arar2018dynamic,wajc2020rounding,bernstein2021framework}, it seems plausible that this characterization will prove useful in subsequent developments in the area.

\smallskip 

Our second contribution is in exhibiting the power of dynamic matching algorithms for restricted class families---bipartite graphs--- to improve algorithms in general graphs. Prior work has relied on fast algorithms for bounded-degree \cite{gupta2013fully} or bounded-arboricity graphs \cite{peleg2016dynamic} to improve their running time \cite{bernstein2015fully,bernstein2016faster,arar2018dynamic,wajc2020rounding,bhattacharya2021deterministic}. 
We present the first use of dynamic  matching algorithms for \emph{bipartite} graphs to improve dynamic matching algorithms in general graphs (beyond the trivial, yet overlooked reduction which incurs a loss factor of $\frac{3}{2}$ in the approximation ratio, discussed in \Cref{sec:bipartite:reduction}). 
This adds to the list of tools for dynamic matching algorithms.

\smallskip 

Finally, our work reinforces the message of prior work, that the study of dynamic algorithms with bounded recourse (number of changes to the output) may be of interest beyond its fundamental nature: reducing recourse may prove useful in reducing algorithms' \emph{update time} (as we show---even for the same problem!). 
This suggests that some slow but recourse-bounded algorithms,  may prove useful in designing fast algorithms for problems for which no such algorithms are known (e.g., \cite{gupta2020fully,bhattacharya2021online}).

\color{black}

\subsection{Related Work}

We briefly review the most relevant algorithmic results for the dynamic matching problem.
For a more detailed survey of the rich literature on dynamic matching and dynamic algorithms more broadly, we refer to the recent survey of \citet{hanauer2021recent}.

\medskip\noindent\textbf{Breaking the Approximation Barrier of Two.} 
For many computational models with dynamic inputs, such as streaming and online algorithms, an approximation ratio of two is easy to achieve for the maximum matching problem, and beating this bound is either a major open problem, or is provably impossible \cite{gamlath2019online}.
In the dynamic graph setting, 
\citet{neiman2016simple} were the first to break the natural approximation ratio of two, giving an $O(\sqrt{m})=O(n)$ worst-case time algorithm with approximation ratio of $\frac{3}{2}$, later improved to $(1+\eps)$  \cite{gupta2013fully}.
\citet{bernstein2015fully,bernstein2016faster} then showed how to maintain a $\frac{3}{2}+\eps$ approximation in $O(\sqrt[4]{m}/\poly(\eps))=O(\sqrt{n}/\poly(\eps))$ time---worst-case for bipartite graphs, and amortized for general graphs. In a recent work, \citet{behnezhad2020fully} showed that for any $\eps>0$, there exists a \emph{randomized} algorithm with worst-case $O(\Delta^\eps) = O(n^\eps)$ update time, and approximation ratio of $2-\frac{1}{1000\cdot 2^{13/\eps}}$, also relying on the search for short augmenting paths. For this, their algorithm crucially relies on the oblivious adversary assumption.
Whether the barrier of $2$ can be broken in $o(n)$ worst-case time by an algorithm (whether deterministic or randomized) which does not require this assumption remained a tantalizing open question, which we resolve in the affirmative.

\medskip\noindent\textbf{Breaking the $O(n)$ Worst-Case Time Barrier.}
As mentioned above, $\Omega(n)$ update time is likely impossible for exact dynamic matching algorithms \cite{dahlgaard2016hardness,henzinger2015unifying,kopelowitz2016higher}.
On the other hand, the dynamic matching literature abounds with (approximate) algorithms breaking the $O(n)$ time barrier \cite{wajc2020rounding,arar2018dynamic,charikar2018fully,bernstein2019deamortization,bhattacharya2021deterministic,bhattacharya2018deterministic,bhattacharya2016new,solomon2016fully,onak2010maintaining}. 
For example, \citet{wajc2020rounding} presented a family of deterministic algorithms trading off worst-case update time and approximation, requiring $\tilde{O}(n^{1/c}\cdot \poly(1/\eps))$ update time for $(2+\eps)\cdot c$ approximation, for any $c\geq 1$ and $\eps>0$. This, however, does not result in $o(n)$ update time for all $2+\eps$ approximation, let alone for $2-\Omega(1)$ approximation. 
The only sub-linear time sub-$2$-approximate algorithms known in general graphs 
are the aforementioned \emph{randomized} algorithm of \citet{behnezhad2020fully} and the deterministic \emph{amortized} time algorithm of \citet{bernstein2016faster}.
Indeed, for $9/4$-approximate matching (or better), no deterministic algorithm for general graphs with $o(n)$ worst-case update time was previously known.
On the other hand, the extent of the usefulness of randomization and amortization are key questions in the dynamic algorithms literature. Fittingly, much effort has been spent on de-randomizing and de-amortizing dynamic matching results (see, e.g.,  \cite{bhattacharya2021deterministic,bernstein2019deamortization}). 
In this work, we show that neither randomization nor amortization are needed to achieve sub-linear-time sub-two-approximate dynamic matching algorithms.


\medskip\noindent\textbf{Subsequent work.} Following the posting of this work, two papers obtaining improved $3/2+\eps$ approximation in deterministic $o(n)$ worst-case update time (respectively, in $O_\eps(\sqrt[4]{m})$ and $O_\eps(\sqrt{n})$) were posted \cite{grandoni2021maintaining,kiss2021improving}. Those later papers' techniques are orthogonal to ours.
\color{black}

\section{Preliminaries}\label{sec:prelims}

\textbf{Problem statement.} 
Our input is a dynamically changing graph $G=(V,E)$, initially empty, undergoing edge \emph{updates} (insertions and deletions). We denote by $G_t$ this graph after $t$ updates.
The objective is to maintain a matching, i.e., a node-disjoint set of edges, of size close to the maximum matching size, $\mu(G)$, spending little time following each update. We denote a maximum (weight) matching in $G$ by $M^*$.
For a weighted graph $G=(V,E,w)$, we wish to compute a matching whose weight approximate the maximum weight matching in $G$, denoted by $MWM(G):=w(M^*):=\sum_{e\in M^*} w(e)$.

\medskip\noindent\textbf{Matching Theory Basics.} 
A matching $M$ is (inclusion-wise) maximal in $G=(V,E)$ if $G$ contains no matching $M'\supsetneq M$.
For a matching $M$ in $G=(V,E)$, an alternating path $P$ in $G$ is a path whose edges alternate between $M$ and $E\setminus M$. An augmenting path is an alternating path whose endpoints are unmatched in $M$. For a set of $k$ disjoint augmenting paths $P_1,P_2,\dots,P_k$ with respect to $M$, the symmetric difference between $M$ and $P:=\bigcup_i P_i$, denoted by $M\bigoplus P$, is a matching of cardinality $|M|+k$, matching all nodes matched by $M$ (and others). The symmetric difference of $M$ and $M^*$ consists of even length paths and cycles, as well as $|M^*|-|M|$ odd-length augmenting paths.
We say an edge $e\in M$ is \emph{3-augmentable} if the connected component of $M\bigoplus M^*$ containing $e$ is an augmenting path of length three.
A standard result in matching theory (see, e.g., \cite[Lemma 1]{konrad2012maximum}) is that a maximal matching which is not much better than $2$-approximate must consist mostly of 3-augmentable edges.

\begin{prop}\label{pre: threeAugmenting}
    Let $\eps \geq 0$. Let $M$ be a maximal matching of $G$ s.t. $|M| \leq \left(\frac{1}{2}+\eps\right)\cdot \mu(G)$. Then $M$ contains at least $\left(\frac{1}{2}-3\eps\right)\cdot \mu(G)$ edges which are 3-augmentable.
\end{prop}

\subsection{Dynamic Matching Background}

In addition to our new structural results and algorithms, our work will rely on a number of previous algorithms from the dynamic matching literature, which we now outline.

\medskip\noindent\textbf{Bounded-Degree Algorithms.} 
One useful property for approximate dynamic matching is the ``stability'' of the matching problem, first used in a dynamic setting by \citet{gupta2013fully}, which implies that for an $\alpha$-approximate matching $M$ in $G_t$, the undeleted edges of $M$ by time $t'\in [t+1, t+\eps\cdot \mu(G_t)]$ constitute an $\alpha(1+\eps)$-approximate matching in $G_{t'}$.
Combined with static linear-time near-maximum matching algorithms \cite{micali1980v,duan2014linear}, this yields near-maximum algorithms with worst-case update time linear in the maximum degree.
This is implied by the work of \citet{gupta2013fully} (first observed in \citet{bernstein2015fully}) and by the arboricity-time algorithm of \citet{peleg2016dynamic}. For completeness, we provide a short proof in \Cref{appendix:prelims}.

\begin{restatable}{prop}{boundeddeg}(Bounded-Degree Algorithm)\label{bounded-deg-algo}
    For any $\eps\leq 1/3$, there exists a deterministic $(1+\eps)$-approximate matching algorithm with worst-case update time $O(\Delta/\eps^2)$ in dynamic graphs with maximum degree at most $\Delta$. 
\end{restatable}

\medskip\noindent\textbf{Kernels.} \Cref{bounded-deg-algo} motivates the study of dynamic \emph{matching sparsifiers}---sparse (specifically, low-degree) subgraphs which approximately preserve the maximum matching in the dynamic graph $G$. This is the approach followed by numerous works in this area \cite{arar2018dynamic,bhattacharya2018deterministic,bernstein2015fully,bernstein2016faster,bernstein2021framework,gupta2013fully,wajc2020rounding,bhattacharya2021deterministic}.
One family of sparsifiers which have proven useful in a number of these works are \emph{kernels}, introduced by \cite{bhattacharya2018deterministic}.

\begin{restatable}{Def}{kerneldef}(Kernels \cite{bhattacharya2018deterministic})\label{def:kernel}
	An \emph{$(\eps,d)$-kernel} of a graph $G=(V,E)$ is an edge-induced subgraph $K=(V,E_K)$ satisfying the following properties:

\begin{enumerate}[label=(\subscript{P}{{\arabic*}}),leftmargin=2\parindent]
		\item \label{p1:bounded-deg} $\max_{v\in V} d_{H}(v)\leq d$.
		\item \label{p2:satisfied-edges} $\max_{v\in e} d_{H}(v) \geq d(1-\eps)$ for each edge $e\in E\setminus E_K$.
	\end{enumerate}
\end{restatable}

By definition, kernels are low-degre subgraphs, and hence one can maintain a near-maximum degree in these graphs in time $O_\eps(d)$, by \Cref{bounded-deg-algo}.
The following proposition implies that such a near-maximum matching in a kernel $K$ is also large in $G$. (See \cite{arar2018dynamic,wajc2020rounding,bhattacharya2018deterministic} for proofs.)

\begin{prop}[The Basic Kernel Lemma]\label{kernel:basic}
    Let $\eps\in (0,1/2)$ and $d\geq 1/\eps$. If $K=(V,E_K)$ is an $(\eps,d)$-kernel of $G=(V,E)$, then 
    $\mu(G)\leq \frac{2(1+\eps)}{1-\eps}\cdot \mu(K) \leq (2+8\eps)\cdot \mu(K).$
\end{prop}

Combining propositions \ref{bounded-deg-algo} and \ref{kernel:basic}, one obtains $(2+O(\eps))$-approximate matching algorithms with update time $O((d/\eps^2)\cdot C(\eps,d)+T(\eps,d))$ from $(\eps,d)$-kernel maintenance algorithms with $T(\eps,d)$ update time and $C(\eps,d)$ changes to the kernel per update to $G$. Our result of \Cref{thm:fast-kernel-algo} follows precisely this well-trodden path.
The key novelty in our work is showing how to surpass the approximation factor of $2+\eps$ obtainable using kernels alone.

\medskip\noindent\textbf{Fast Bipartite Matching.} 
In our main algorithm we dynamically maintain large matchings in dynamically changing bipartite subgraphs of $G$, using the following result of \citet{bernstein2015fully}.

\begin{prop}\label{edcs-algo}(Dynamic Bipartite Matching)
    Let $\eps>0$. There exists a $(3/2+\eps)$-approximate dynamic bipartite matching algorithm with worst-case update time $O(\sqrt[4]{m}\cdot \poly(1/\eps))$.
\end{prop}

\medskip\noindent\textbf{Bounding Recourse.}
Recall that to dynamize our main approach, we need that the augmenting edges $A$, which are the union of large matching $M_i$ in appropriately-chosen bipartite subgraphs $B_i$, must change slowly.
That is, the number of changes to these $M_i$ per update, also referred to as \emph{recourse}, should be small. For this, we rely on the following (minor extension) of the recent black-box reduction of \citet[Theorem 3]{solomon2021generalized}, which we prove in \Cref{appendix:prelims} for completeness.

\begin{restatable}{prop}{recourse}(Bounded Recourse) \label{recourse}
For any $\alpha\geq 1$ and $\eps\in (0,1/6)$, if there exist an $\alpha$-approximate dynamic matching algorithm with worst-case update time $T$,
then there exists an $\alpha(1 + \eps)$-approximate dynamic matching algorithm with worst-case update time $T + O(1/\eps)$ and worst-case recourse bound $O(1/\eps)$. This applies under edge and/or node updates \textbf{and/or star updates.}
\end{restatable}
\citet{solomon2021generalized} note that in some settings, bounded recourse translates to bounded running time. Our work provides another concrete example of this phenomenon, in a dynamic setting.

\section{Identifying Augmentations to Kernels}\label{sec:kernels++}

In this section we characterize a bounded-degree subgraph $A$ whose addition to $K$ results in a better-than-two-approximate matching sparsifier, $AK:=A\cup K$. In following sections we will show that this subgraph $AK$ is also efficiently maintainable.
To focus on the key ideas, we mostly provide brief proof sketches here, deferring most full proofs with detailed calculations to \Cref{appendix:kernel++}.

Our starting point, in \Cref{sec:extended-lemma}, is an extension of the proof of the approximation ratio of kernels, to obtain bounds on degrees in $K$ of endpoints of any maximum matching $M^*$ in $G$, whenever such $K$ is worse than $(2-\delta)$-approximate.
In \Cref{sec:characterizing} we 
leverage these degree bounds to characterize (most) augmenting paths of some maximum matching $M$ in $K$ in terms of the $K$-degree of their nodes.
In \Cref{sec:two-matchings-suffice} we use this characterization to prove that large matchings in two bipartite subgraphs, defined by degrees in $K$ (broadly: low-degree and high-degree nodes on opposite sides), are precisely the desired set $A$.
Unfortunately, defining these subgraphs based on degrees in $K$, which may change abruptly, causes updates to $G$ (and hence to $K$) to result in \emph{node updates} in these subgraphs, requiring $\Omega(n)$ update time.
In \Cref{sec:main-approx-degrees} we show that large matchings in some $O(1)$ bipartite subgraphs defined by \emph{approximate} degrees in $K$ (causing these graphs to only change by small star updates) still result in a sufficient $\Omega(\delta \cdot \mu(G))$ augmenting paths with respect to $M$.

\medskip\noindent\textbf{Notation and parameters.}
In this section we fix an $(\eps,d)$-kernel $K=(V,E_H)$ of graph $G=(V,E)$ with $\eps = 2\times 10^{-8}$ and $d\geq 1/\eps$.
We also fix two additional parameters $\delta = 2\times 10^{-6}$ and $s = 2 \times 10^{-4}$.
We assume that $\mu(G) \geq (2-\delta)\cdot \mu(K)$, and search for a sparse augmentation $A$ of $K$, resulting in a sparse subgraph $AK:=A\cup K$ with $\mu(G)\leq (2-\delta)\cdot \mu(AK).$
Our choices of values for $\eps,d,s$ and $\delta$ are taken to make these numbers rather round, while allowing $\delta$ to be close to its highest possible value. 
We indicate by $f(\eps,s,\delta) \leqstar g(\eps,s,\delta)$ inequalities which hold for sufficiently small $\eps$, $s$, $\delta$, $\eps/\delta$ and $\delta/s$, and which specifically hold for our particular choices of $\eps,s$ and $\delta$.

\subsection{The Extended Kernel Lemma}\label{sec:extended-lemma}

By \Cref{kernel:basic}, kernels are $(2+O(\eps))$-approximate sparsifiers. In this section we prove that for a kernel $K$ to   be not  much better than a $2$-approximation, most edges of any maximum matching $M^*$ in $G$ must have their endpoint degrees in $K$ sum to roughly $d$. More precisely, we prove the following.

\begin{restatable}{lem}{ekl}(Extended Kernel Lemma)\label{lem: kernelLemma} 
For any maximum matching $M^*$ in $G$, at most $\frac{\delta}{s}\cdot |M^*|$ edges $e\in M^*$ satisfy either of the following conditions.
\begin{enumerate}
    \item \label{high-degree-sum-condition} $\sum_{v\in e} d_K(v) \geq d\cdot (1+s-2\eps)$.
    \item \label{low-degree-sum-condition}$\sum_{v\in e} d_K(v) \leq d\cdot (1-s+2\eps)$ and $e\in E_K$.
\end{enumerate}
\end{restatable}

We now outline the proof approach, deferring a full proof with detailed calculations to  \Cref{appendix:kernel++}.

\begin{proof}[Proof (Sketch)]
    The basic kernel lemma, \Cref{kernel:basic}, which this lemma extends, is obtained by considering a fractional matching $\vec{x}$ in $K$, with values determined by a maximum matching $M^*$ in $G$. 
    For this fractional matching $\vec{x}$, nodes' fractional degrees, $y_v:=\sum_{e\ni v} x_e$, satisfy the following bound. 
    \begin{align}\label{kernel-lem-deg-bound}
    \sum_{v\in e} y_v\geq 1-\eps \qquad \forall e\in M^*.
    \end{align}
    This bound directly implies that this fractional matching $x$ has large value in terms of $|M^*|=\mu(G)$.
    \begin{align}\label{kernel-lem-size-bound}
    \sum_e x_e = \frac{1}{2}\sum_v y_v \geq \frac{1}{2}\cdot \sum_{e\in M^*} \sum_{v\in e} y_v \geq \frac{1}{2}\cdot (1-\eps)\cdot \mu(G).
    \end{align}
    In bipartite graphs $G$, where the fractional matching polytope is integral, this immediately implies that $\mu(K)\geq \sum_e x_e\geq \frac{1-\eps}{2}\cdot \mu(G)$.
    In general graphs, using an application of Vizing's edge coloring theorem (see \Cref{vizing-app}), one can show that for this particular fractional matching $\vec{x}$, the kernel $K$ contains a large integral matching $M$ of nearly the same size as $x$'s value, namely 
    \begin{align}\label{vizing-app-internal}
    \mu(K)\geq |M|\geq \frac{1}{1+1/d}\cdot \sum_e x_e\geq \frac{1}{1+\eps}\cdot x_e.
    \end{align}
    Combining equations \eqref{kernel-lem-size-bound} and \eqref{vizing-app-internal} directly yields \Cref{kernel:basic}.
    
    \textbf{Our key observation} is that the fractional degree bound of \Cref{kernel-lem-deg-bound} is loose if either condition \eqref{high-degree-sum-condition} or \eqref{low-degree-sum-condition} hold for edge $e\in M^*$. For such edge $e$ we have that $\sum_{v\in e} y_v \geq 1+s-4\eps$. 
    Consequently, if $r$ is the number of edges of $M^*$ satisfying either condition, then the above observation and \Cref{kernel-lem-deg-bound} imply that the fractional matching $\vec{x}$ has size at least $\sum_e x_e \geq \frac{1-\eps}{2} \cdot \mu(G) + r\cdot (s-3\eps)$. By \Cref{vizing-app-internal}, an increase in $r$ therefore increases our lower bound on $\mu(K)$ in terms of $\mu(G)$. Thus, the inequality $\mu(G)\geq (2-\delta)\cdot \mu(K)$, and some simple calculations, implies the claimed upper bound on $r$.
\end{proof}

\subsection{Characterizing augmenting path node classes}\label{sec:characterizing}
In this section, we characterize the structure of length-three augmenting paths based on node  degree classes, which we now define. 

\begin{Def} \label{def: nodes-class}
A node $v$ has \emph{super-high/high/medium/low} degree in $K$ for $i \in [\frac{1}{\eps}]$, if it belongs to the following sets, respectively.\footnote{For notational simplicity, we assume that $d$ is an integer multiple of $1/\eps$. Since $d=\geq 1/\eps$, this is WLOG, up to rounding of $d$, which at most increases our running time by an $O(1)$ factor.}
\begin{align*}
    V^{(i)}_{SH} & := \left\{v\in V \mid d_K(v)\geq d\cdot\left(1-\eps- i \eps^2\right)\right\} \\
    V^{(i)}_{H} & := \left\{v\in V \mid d_K(v)\geq d\cdot \left(1-2s - i \eps^2 \right)\right\} \\
    V^{(i)}_{M} & := \left\{v\in V \mid d_K(v)\in \left(d\cdot \left(s+ i \eps^2 \right) , d\cdot \left(1-2s - i \eps^2\right)\right)\right\} \\
    V^{(i)}_{L} & := \left\{v\in V \mid d_K(v)\leq d\cdot \left(s+i \eps^2\right)\right\}.
\end{align*}
\end{Def}

The proofs of sections \ref{sec:characterizing} and \ref{sec:two-matchings-suffice} will hold for any $i\in [\frac{1}{\eps}]$. Therefore, for notational simplicity, we fix some such $i$, and drop the notation $(i)$ from the node's classes. The reader will be relieved to know that they need not commit to memory the exact formulae defining the sets $V_{SH},V_H,V_M$ and $V_L$. In particular, until \Cref{sec:main-approx-degrees}, these formulae will only be used explicitly in the following two observations and lemma.

\begin{obs}\label{obs: disjoint-classes}
    We have that $V_{SH}\subseteq V_H$. Also, $V_H, V_M$, and $V_L$ are disjoint.
\end{obs}

Next, we note that \Cref{lem: kernelLemma} implies constraints on degree classes of endpoints of $M^*$ edges.

\begin{obs} \label{obs: edge-restriction}
At most $\frac{\delta}{s}\cdot\mu(K)$ edges $(u,v)\in M^*$ satisfy either of the following conditions.
\begin{enumerate}
    \item $(u,v) \in (V_M \times V_{SH}) \cup (V_H \times V_H)$.\label{high-sum}
    \item $(u,v)\in V_L \times (V_L \cup V_M)$. \label{low-sum}
\end{enumerate}
\end{obs}
\begin{proof}
    Condition \eqref{high-sum} is a special cases of Condition \eqref{high-degree-sum-condition} of \Cref{lem: kernelLemma}.  On the other hand, every edge $(u,v)\in V_L\times (V_L\cup V_M)$ has no super-high degree endpoint, and so by Property \ref{p2:satisfied-edges} of a kernel, must belong to $E_K$. Therefore, Condition \eqref{low-sum} is a special case of Condition \eqref{low-degree-sum-condition} of \Cref{lem: kernelLemma} too. The bound on the number of edge in $M^*$ satisfying either condition above therefore follows from \Cref{lem: kernelLemma}.
\end{proof}

Next, we note that some maximum matching $M$ in $K$ matches most $V_H$ nodes.

\begin{restatable}{lem}{Hfree}\label{lem: Hnodes-matches}
Some maximum matching $M$ in $K$ leaves at most $7s \cdot \mu(K)$ nodes of $V_H$ unmatched.
\end{restatable}
\begin{proof}[Proof (Sketch)]
    First, we note that there exists a (possibly non-maximum) matching $M'$ in $K$ which leaves at most $O(s)\cdot \mu(K)$ high-degree nodes in $K$ unmatched. This follows by considering a $(d+1)$-edge-coloring of $K$ (guaranteed to exist by $\Delta(K)\leq d$ and Vizing's theorem), and noting that a high degree node $v$ in such a color is matched with probability at least $d_K(v)/(d+1)\geq 1-O(s)$. Linearity of expectation then implies a bound of $O(s)\cdot |V_H|$ on the number of nodes in $V_H$ unmatched by $M'$. However, $|V_H|$ can be shown to be $O(\mu(K))$, and so we find that $M'$ leaves at most $O(s)\cdot \mu(K)$ unmatched nodes in $V_H$. The proof is completed by augmenting $M'$ to be a maximum matching in $M$ while only adding to the set of matched nodes.
\end{proof}

For the remainder of \Cref{sec:kernels++}, we let $M$ be the maximum matching in $K$ guaranteed by \Cref{lem: Hnodes-matches}. That lemma, and Property \ref{p2:satisfied-edges} of kernel $K$ implies that $M$ is an \emph{$O(s)$-approximately maximal matching}, i.e., it is a matching which is maximal in a subgraph obtained by removing some $O(s)\cdot \mu(G)$ nodes (and their edges). These type of matchings were useful for \citet{peleg2016dynamic} when studying bounded-arboricity graphs, and will prove useful for us, too. In particular, since such a matching is maximal in a graph $G[V\setminus U]$ with essentially the same matching size as $G$, we can appeal to \Cref{pre: threeAugmenting} to argue that $M$ has many 3-augmentable edges. (See \Cref{appendix:kernel++}.)

\begin{restatable}{lem}{manyaug}\label{lem: manyLength3}
     At least $(1 - 88s)\cdot \mu(K)$ edges of $M$ are 3-augmentable.
\end{restatable}

In what follows, we will show that most length-three augmenting paths in $M \cup M^*$ belong to one of a small number of \emph{types}, which we now define. Note that since $M$ is maximum in $K$, it admits no augmenting paths in $E_K$, and hence an augmenting path $v_1-v_2-v_3-v_4$ cannot have both of its $M^*$ edges $(v_1,v_2)$ and $(v_3,v_4)$) in $E_K$. Therefore, we assume WLOG that $(v_3,v_4)\not\in E_K$.

\begin{Def} \label{def: freq-aug}
We say a length-three augmenting path $p:v_1-v_2-v_3-v_4$ with edges $e_i=(v_i,v_{i+1})$ for $i\in [3]$ and $e_3=(v_3,v_4)\not\in E_K$ is \emph{frequent} if it belongs to one of the following \emph{types} (see \Cref{aug-paths}):
\begin{itemize}
    \item Type 1: $(v_1,v_2,v_3,v_4)\in V_{L}\times V_{SH} \times V_{SH} \times V_L$.
    \item Type 2: $(v_1,v_2,v_3,v_4)\in V_{L}\times (V_{H} \setminus V_{SH}) \times V_{SH} \times V_L$ and $e_1\in E_K$.
    \item Type 3: $(v_1,v_2,v_3,v_4)\in V_{M}\times (V_{H} \setminus V_{SH}) \times V_{SH} \times V_L$ and $e_1\in E_K$.
    \item Type 4: $(v_1,v_2,v_3,v_4)\in V_{M}\times V_{M} \times V_{SH} \times V_L$ and $e_1\in E_K$.
    \end{itemize}
\end{Def}

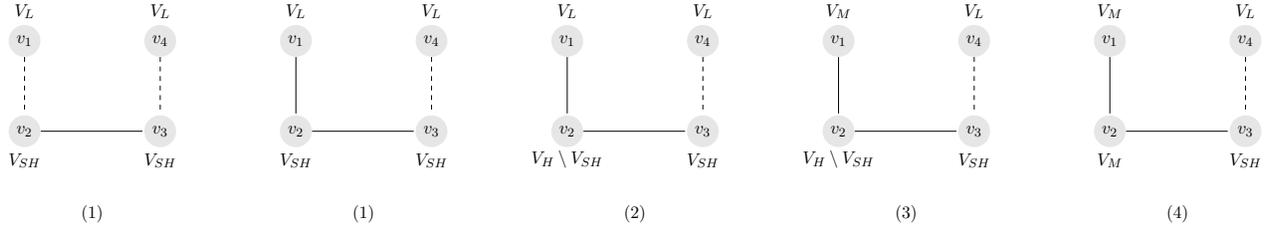
\begin{figure}[H]
  \centering

  \resizebox{1\textwidth}{!}{
  \begin{tikzpicture}[shorten >=1pt,-]
    \tikzstyle{node}=[circle,fill=black!10,minimum size=20pt,inner sep=0pt];

    \node[node, label={$V_L$}] (1-1) at (3, 4) {$v_1$};
    \node[node, label={$V_L$}] (1-2) at (6, 4) {$v_4$};
    \node[node, label={[yshift=-1.3cm]$V_{SH}$}] (2-1) at (3, 2) {$v_2$};
    \node[node, label={[yshift=-1.3cm]$V_{SH}$}] (2-2) at (6, 2) {$v_3$};
    \draw (2-1) -- node[yshift=-1.5cm,below] {$(1)$} ++ (2-2) ;
    \draw[dashed] (1-1) -- (2-1);
    \draw[dashed] (1-2) -- (2-2);

    \node[node, label={$V_L$}] (1-3) at (9, 4) {$v_1$};
    \node[node, label={$V_L$}] (1-4) at (12, 4) {$v_4$};
    \node[node, label={[yshift=-1.3cm]$V_{SH}$}] (2-3) at (9, 2) {$v_2$};
    \node[node, label={[yshift=-1.3cm]$V_{SH}$}] (2-4) at (12, 2) {$v_3$};
    \draw (2-3) -- node[yshift=-1.5cm,below] {$(1)$} ++ (2-4) ;
    \draw (1-3) -- (2-3);
    \draw[dashed] (1-4) -- (2-4);
    
    \node[node, label={$V_L$}] (1-5) at (15, 4) {$v_1$};
    \node[node, label={$V_L$}] (1-6) at (18, 4) {$v_4$};
    \node[node, label={[yshift=-1.3cm]$V_{H} \setminus V_{SH}$}] (2-5) at (15, 2) {$v_2$};
    \node[node, label={[yshift=-1.3cm]$V_{SH}$}] (2-6) at (18, 2) {$v_3$};
    \draw (2-5) -- node[yshift=-1.5cm,below] {$(2)$} ++ (2-6) ;
    \draw (1-5) -- (2-5);
    \draw[dashed] (1-6) -- (2-6);
    
    \node[node, label={$V_M$}] (1-7) at (21, 4) {$v_1$};
    \node[node, label={$V_L$}] (1-8) at (24, 4) {$v_4$};
    \node[node, label={[yshift=-1.3cm]$V_{H} \setminus V_{SH}$}] (2-7) at (21, 2) {$v_2$};
    \node[node, label={[yshift=-1.3cm]$V_{SH}$}] (2-8) at (24, 2) {$v_3$};
    \draw (2-7) -- node[yshift=-1.5cm,below] {$(3)$} ++ (2-8) ;
    \draw (1-7) -- (2-7);
    \draw[dashed] (1-8) -- (2-8);

    \node[node, label={$V_M$}] (1-9) at (27, 4) {$v_1$};
    \node[node, label={$V_L$}] (1-10) at (30, 4) {$v_4$};
    \node[node, label={[yshift=-1.3cm]$V_M$}] (2-9) at (27, 2) {$v_2$};
    \node[node, label={[yshift=-1.3cm]$V_{SH}$}] (2-10) at (30, 2) {$v_3$};
    \draw (2-9) -- node[yshift=-1.5cm,below] {$(4)$} ++ (2-10) ;
    \draw (1-9) -- (2-9);
    \draw[dashed] (1-10) -- (2-10);
    
  \end{tikzpicture}
  }
\caption{Frequent length-three augmenting paths.}
\subcaption{Subfigures are labeled by their path's type number. (Note the two type (1) paths.)\\
	Kernel edges are solid and non-kernel edges are dashed. Nodes are labeled by their class.}
\label{aug-paths}
\end{figure}

We now prove that these path types are indeed frequent, as their name suggests. Let $n_f$ and $n_{if}$ denote the number of frequent and infrequent (i.e., not frequent) length-three augmenting paths in $M \cup M^*$, respectively. We prove the following.

\begin{lem}\label{lem: manyNf}
    We have that $n_{if}\leq \left(7s+\frac{\delta}{s}\right)\cdot \mu(K)$ and $n_f\geq (1-100s-\frac{\delta}{s})\cdot \mu(K)$.
\end{lem}

\begin{proof}
In what follows, whenever discussing a length-three augmenting path, we denote it by $p:v_1-v_2-v_3-v_4$, and denote its edges by $e_i=(v_i,v_{i+1})$ for $i \in [3]$. Since $M$ is a maximum matching in $K$, we cannot have $e_1,e_3\in E_K$, since the converse would entail an augmenting path for $M$ in $K$, contradicting $|M|=\mu(K)$. So, we assume WLOG that $e_3 \not\in E_K.$
We start by upper bounding the number of infrequent paths, first considering them by their first node, $v_1$.

By \Cref{obs: edge-restriction}, at most $\frac{\delta}{s}\cdot \mu(K)$ length-three augmenting paths are infrequent paths due to $(v_1,v_2)\in V_L \times (V_L \cup V_M)$ or $(v_1,v_2)\in V_M\times (V_L\cup V_{SH})$.
Now, by \Cref{lem: Hnodes-matches}, at most $7s\cdot\mu(K)$ length-three augmenting paths have a $V_H$ node as one of their endpoints (i.e., their nodes $v_1,v_4$, which are unmatched in $M$). 
This accounts for infrequent paths with $v_1\in V_H$, as well infrequent paths with $v_4\in H$.
In all remaining paths, since $e_3 = (v_3,v_4)\notin E_K$, we have by Property \ref{p2:satisfied-edges} of the kernel that $v_3\in V_{SH}$, and $v_4\in V_L$ (the case $(v_3,v_4)\in V_{SH}\times (V \setminus V_{L})$ is accounted for by the at most $\frac{\delta}{s}\cdot \mu(K)$ edges satisfying conditions \ref{high-sum} or \ref{low-sum} of \Cref{obs: edge-restriction}).
All remaining paths with $(v_3,v_4)\in (V_{SH}\times V_L)\setminus E_K$ correspond to frequent paths of types (1), (2), (3) and (4), since they satisfy, respectively, 
$$(v_1,v_2)\in  (V_L\times V_{SH}) \cup (V_L\times (V_H\setminus V_{SH})) \cup (V_M \times (V_H\setminus V_{SH})) \cup (V_M\times V_M) .$$

To summarize, the number of infrequent length-three augmenting paths is at most $n_{if} \leq \left(7s + \frac{\delta}{s}\right) \cdot \mu(K)$.
By \Cref{lem: manyLength3}, the number of length-three augmenting paths is at least $(1-88s)\cdot \mu(K)$. Combined, these two bounds imply the second claim, i.e., $n_f\geq (1-95s)\cdot \mu(K)\geq (1-100s)\cdot \mu(K)$.
\end{proof}

By \Cref{lem: manyNf}, many edges of length-three augmenting paths in $M\cup M^*$ have 
a (possibly super-)high-degree node as 
one endpoint, and a low-degree node as 
the other endpoint. This motivates us to compute large bipartite matchings between high- and low-degree nodes in $K$. We address this strategy in the following section.

\subsection{Two bipartite matchings suffice}\label{sec:two-matchings-suffice}
In this section, we prove that large matchings in judiciously-chose bipartite subgraphs contain many augmenting paths. We start by noting that one such naturally-chosen matching does not yield such augmenting paths.

\medskip\noindent\textbf{One matching does not suffice.}
Let $G[A, B]:=G[A\times B]$ denote the bipartite subgraph of $G$ induced by bipartition $(A,B)$.
The previous section's characterization of most augmenting paths suggests that a large matching in $G[V_{SH}, V_L]$ together with edges of $K$ might contain many augmenting paths with respect to $M$. However, this is not the case. For example, if all augmenting paths are of type (2), then matching all copies of $v_3$ to $v_1$ in the same path (the unmatched $V_L$ neighbor of $v_2\in V_H$) would not result in \emph{any} augmenting paths.
This example might suggest to instead compute a large matching in the subgraph $G[V_H, V_L]$. However, a similar problem then arises in paths of type (3).
\smallskip 

To conclude, a large matching in either $G[V_{SH},V_L]$ or $G[V_H,V_L]$ may not constitute an augmentation $A$ of the kernel $K$.
We will prove that a large matching in $G[V_{SH},V_L]$ and a large matching in $G[V_H,V_L]$ \emph{together} form the desired augmentation $A$ of the kernel. Formally, we prove the following.

\begin{lem}\label{lem: augmentingKernel}
    Let $M_H$ and $M_{SH}$ be a $(3/2 + \eps)$-approximate matching in $G[V_H, V_L]$ and $G[V_{SH}, V_L]$, respectively. If $A(M_H)$ and $A(M_{SH})$ denote the minimum number of disjoint augmenting paths induced by the union of $K$ and each of $M_H$ and $M_{SH}$, respectively, then $\max \left(A(M_H), A(M_{SH})\right) \geq (\delta + \eps)\cdot \mu(K)$.
\end{lem}

The proof of the above lemma will require some setting up, starting with the following definition.

\begin{Def}\label{def: badNode}
Node $v$ is \emph{bad} if $v \in V(M)$ and $v$ is not in a frequent length-three augmenting path.
\end{Def}
\begin{obs}\label{low-good}
    If node $v\in V_L$ is not bad, then it is unmatched in $M$.
\end{obs}

\begin{obs}\label{lem: fewBadNodes}
     The number of bad nodes is at most $b\leq (200s + \frac{2\delta}{s})\cdot \mu(K)$.
\end{obs}
\begin{proof}
First, by \Cref{lem: manyLength3}, there are at most $88s\cdot\mu(K)$ edges of $M$ that are not 3-augmentable. On the other hand, by \Cref{lem: manyNf} we have that $n_{if} \leq (7s + \frac{\delta}{s})\cdot\mu(K)$.
Combining both bounds on the number of edges containing a bad node, and noting that each such edge contributes two bad nodes, we have that $b \leq 2\cdot(88s+7s+\frac{\delta}{s})\cdot \mu(K) \leq (200s+\frac{\delta}{s})\cdot \mu(K)$, as claimed.
\end{proof}

We now lower bound the number of augmemting paths in $M_H\cup E_K$ and $M_{SH}\cup E_K$ in terms of numbers of frequent augmenting paths of the various types.

\begin{lem}\label{aug-lb-in-terms-of-n_i}
    Let $M_H,M_{SH},A(M_H)$ and $A(M_{SH})$ be as in \Cref{core-approx-lem}. 
    Let $n_i$ the number of frequent augmenting paths of type $i\in [4]$ (noting that $n_f=n_1+n_2+n_3+n_4$). 
    Then, 
    \begin{align*}
        A(M_{H}) & \geq |M_{H}| - b - n_f.\\
        A(M_{SH}) & \geq |M_{SH}| - b - (n_1 + n_2).
    \end{align*}
\end{lem}

\begin{proof}
    Whenever discussing a frequent length-three augmenting path in this proof, we denote it by $v_1-v_2-v_3-v_4$, as in \Cref{def: freq-aug}.
    We say an edge in $M_H$ or $M_{SH}$ is \emph{good} if neither of its endpoints is bad. 
    Consequently, we have at least $|M_H|-b$ and $|M_{SH}|-b$ good edges in $M_H$ and $M_{SH}$, respectively.
    We say a node is \emph{free} if it is not matched in $M$. 
    By \Cref{low-good}, for any good edge $e\in M_H\cup M_{SH}$, the endpoint $v\in e\cap V_L$ must be free.
    We will show that these good edges form part of large sets of disjoint augmenting paths in $M_H$ and $M_{SH}$.
    
    We first consider the good edges of matching $M_H$, which by definition are between a free node in $V_L$ and a $V_H$ node which is either free or belongs to a frequent length-three augmenting path.
    Any edge in $M_H$ between a free $V_H$ node and a free $V_L$ node constitutes an augmenting path of length one. Also, if two $V_H$ nodes $v_2,v_3$ of the same frequent path of types (1), (2), or (3) match to free $V_L$ nodes, $u_2,u_3$, then $u_2-v_2-v_3-u_3$ is an augmenting path in $M_H\cup E_K$.
    Finally, if the unique $V_{SH}$ node $v_3$ of a frequent path $v_1-v_2-v_3-v_4$ of type (4) is matched to a free node $u$ in $V_L$, then $v_1-v_2-v_3-u$ is an augmenting path in $M_H\cup E_K$.
    (See \Cref{fig: lemma39-proof} for an illustration of some of the above cases.)
    We note further that the augmenting paths above are node disjoint, since no $V_H$ or $V_L$ node is matched twice, and therefore none of these nodes appear in two of these augmenting paths, and moreover, the node $v_1$ in paths of type (4) is not matched in $M_H$, since it belongs to $V_M$, which by \Cref{obs: disjoint-classes} is disjoint from $V_L$ and $V_H$.
    We conclude that any more than one good match per path of type (1), (2) or (3) and any other kind of good match all contribute one augmenting path to this disjoint set of augmenting paths.
    Therefore, the maximum number of disjoint augmenting paths in $M_H \cup E_K$ is, as claimed, at least
\begin{align*}
    A(M_H) \geq |M_H| - b - (n_1 + n_2 + n_3) \geq |M_H| - b - n_f.
\end{align*}

\begin{figure}[h]
  \centering

  \resizebox{1\textwidth}{!}{
    \input{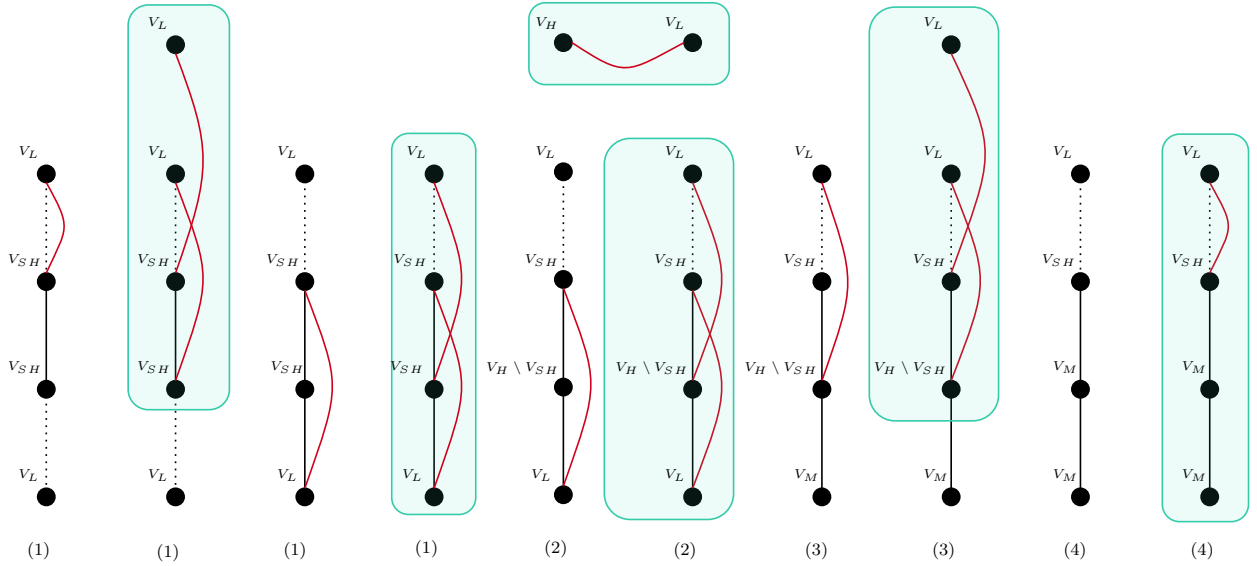}
  }

\caption{(Non-exhaustive) options for matches in $M_H$ and relation to $A(M_H)$.}
\subcaption{Kernel edges are solid and non-kernel edges are dashed, while edges in $M_H$ are curved and red. \\ 
	The number below each vertical line corresponds to the path type of the lowest four nodes in the line.\\
	The blue boxes show augmenting paths in $E_K\cup M_H$}
\label{fig: lemma39-proof}
\end{figure}

    We now consider the good edges of matching $M_{SH}$, only outlining the differences compared to the analysis for $M_{H}$. 
    For matching $M_{SH}$, the only good matches which are guaranteed to form an augmenting paths are: (i) a good match between a free $u\in V_{SH}$ and a free $v\in V_L$ node, for which $(u,v)$ is an augmenting path of length one, and (ii) matches between two $V_{SH}$ nodes $v_2,v_3$ of the same frequent path of type (1) to free $V_L$ nodes, $u_2,u_3$, for which $u_2-v_2-v_3-u_3$ is an augmenting path.
    By similar reasoning to our study of $M_{H}$, we therefore conclude that the maximum number of disjoint augmenting paths in $M_{SH} \cup E_K$ is, as claimed, at least
    \begin{align*}
        A(M_{SH}) & \geq |M_{SH}| - b - (n_1 + n_2).\qedhere
    \end{align*}
\end{proof}

Armed with \Cref{aug-lb-in-terms-of-n_i}, we are now ready to prove that the matchings $M_H$ and $M_{SH}$ form a good augmentation of $K$. That is, we can now prove \Cref{lem: augmentingKernel}.

\begin{proof}[Proof of \Cref{lem: augmentingKernel}]
By \Cref{lem: fewBadNodes}, the number of bad nodes is at most $b \leq (200s + \frac{2\delta}{s})\cdot \mu(K)$, and by \Cref{lem: manyNf}, the number of frequent length-three augmenting paths is at least $n_f\geq (1-100s-\frac{\delta}{s})\cdot \mu(K)$. However, in this proof, we will use the following looser bounds, 
$b \leq (200s + \frac{2\delta}{s} + 16\eps)\cdot \mu(K)$ and
$n_f\geq (1-100s-\frac{\delta}{s} - 8\eps)\cdot \mu(K)$. Our use of these looser bounds will prove useful in \Cref{sec:main-approx-degrees}. 

We lower bound $|M_H|$ and $|M_{SH}|$, by first lower bounding the maximum matching sizes in the bipartite graphs $G[V_H,V_L]$.
The matchings obtained by considering all $V_H\times V_L$ and $V_{SH}\times V_L$ edges in frequent length-three augmenting paths immediately imply that $\mu(G[V_H, V_L]) \geq n_f + n_1 + n_2$, and similarly $\mu(G[V_{SH}, V_L]) \geq n_f + n_1$. 
Here too, we will use the looser bounds $\mu(G[V_H, V_L]) \geq n_f + n_1 + n_2$, and similarly $\mu(G[V_{SH}, V_L]) \geq n_f + n_1$. 
On the other hand, $M_H$ and $M_{SH}$ are $(\frac{3}{2}+\eps)\leq \frac{1}{2/3-\eps}$-approximate matchings in these bipartite subgraphs. Combining the above with \Cref{aug-lb-in-terms-of-n_i}, we therefore get the following bounds on $A(M_H)$ and $A(M_{SH})$.
\begin{align*}
    A(M_H) & \geq \left(\frac{2}{3} - \eps\right)\cdot(n_1 + n_2 + n_f) - n_f - b 
    \geq \frac{2}{9}\cdot n_1 + \frac{2}{3}\cdot n_2 - \frac{1}{3}\cdot n_f - b
    ,
\end{align*}
\begin{align*}
    A(M_{SH}) & \geq \left(\frac{2}{3} - \eps\right)\cdot (n_1 + n_f) - n_1 - n_2 - b \geq \frac{2}{3}\cdot n_f - \frac{1}{3}\cdot n_1 - n_2 - b.
\end{align*}

Therefore, we have that
\begin{align*}
    A(M_H) + A(M_{SH}) & \geq A(M_H) + \frac{2}{3}\cdot A(M_{SH}) \geq \frac{1}{9}\cdot n_f - \frac{5}{3}\cdot b.
\end{align*}

Combining the above with lemmas \ref{lem: manyNf} and \ref{lem: fewBadNodes}, we obtain the desired inequality
\begin{align*}
    \max\left(A(M_H), A(M_{SH})\right) & \geq \frac{1}{2}\cdot (A(M_H) + A(M_{SH})) \\ 
    & \geq \frac{1}{2}\cdot \left(\frac{1}{9}\cdot \left(1-100s-\frac{\delta}{s} - 8\eps \right) - \frac{5}{3}\cdot \left(200s + \frac{2\delta}{s} + 16\eps\right)
    \right)\cdot \mu(K) \\
    & \geqstar (\delta + \eps)\cdot\mu(K). \qedhere
\end{align*}
\end{proof}

\subsection{Approximate degrees in the kernel}\label{sec:main-approx-degrees}
In the previous sections we show that the union of the kernel, a large matching in $G[V_{SH}, V_L]$, and a large matching in $G[V_H, V_L]$ contains a $(2-\delta)$-approximate matching of $G$. However, just \emph{maintaining subgraph} $G[V_{SH}, V_L]$ and $G[V_H, V_L]$ is costly and takes $\Omega(n)$ time to update in the dynamic setting as nodes join and leave these subgraphs when their degrees in $K$ changes. In order to have less abrupt changes, we use approximate degrees for classifying nodes.
\begin{Def}\label{def:approx-degs}
Let $\alpha \geq 0$. Elements of set $\{d^u(v)\mid u,v\in V\}$ are \emph{$\alpha$-approximate degrees} in $K$ if for every pair of nodes $u,v\in V$,  
$$d_K(v) - \alpha \leq d^u(v)\leq d_K(v) +\alpha.$$
\end{Def}

In \Cref{sec: approx-deg-maintenance}, we show how to maintain $\alpha$-approximate degrees in $K$ efficiently. 
For now, we will show that large matchings in bipartite subgraphs defined by $\eps^2d$-approximate degrees contain many disjoint augmenting paths, mirroring  \Cref{lem: augmentingKernel} for bipartite subgraphs based on exact degrees in $K$.
Specifically, we will consider the following bipartite subgraphs.

\begin{Def}\label{def: threshold}
    Let $\{d^u(v) \mid u,v\in V\}$ be $\eps^2 d$-approximate degrees in $K$.
    For $i \in [\frac{1}{\eps}]$, we denote by $B_H^{(i)}:=(V,E_H^{(i)})$ and $B_{SH}^{(i)}:=(V,E_{SH}^{(i)})$ bipartite subgraphs of $G$ induced by the following edge sets.
\begin{alignat*}{3}
E_H^{(i)} & = \left\{(u,v) \mid d^u(v) \geq d\cdot\left(1-2s - i \eps^2\right), \,\,d^v(u) \leq d \cdot \left(s+i \eps^2\right) \right\}.\\
E_{SH}^{(i)} & = \left\{(u,v) \mid d^u(v) \geq d \cdot (1-\eps-i\eps^2), \,\quad d^v(u) \leq d\cdot \left(s+i \eps^2\right) \right\}.   
\end{alignat*} 
\end{Def}

We recall that (informally) $\eps\ll s\ll 1$, and so the above subgraphs are indeed bipartite.

As in previous sections, we will drop the superscript $(i)$ when $i$ is clear from context.
Intuitively, subgraphs $B_H$ and $B_{SH}$ are similar to the bipartite subgraphs $G[V_H, V_L]$ and $G[V_{SH}, V_L]$, respectively. However, since we use approximate degrees, these subgraphs are not equal to each other, since a node's degree class can be incorrectly classified by some of its neighbors, as in the following definition.

\begin{Def}
Fix $i\in [\frac{1}{\eps}]$. Let $v\in S$, where $S\in \{V_{SH},V_H\setminus V_{SH},V_M,V_L\}$.  We say $v$ is \emph{misclassified for $i$} if some neighbor $u$ of $v$ has $d^u(v) = d'$, but if we had $d_K(v)=d'$, then we would have $v\not\in S$.
\end{Def}

Intuitively, misclassifications may result a large number of matches in $M_H$ and $M_{SH}$ which do not increase $A(M_H)$ and $A(M_{SH})$ compared to the exact degree case. That is, such misclassifications result in a lower number of augmenting paths in the exact degree case.
However, as we show, for some $i^*\in [\frac{1}{\eps}]$, such problematic miscalssifications are rare.

\begin{obs}\label{lem: threshold}
Let $\{d^u(v)\mid u,v\in V\}$ be $\eps^2d$-approximate degrees in kernel $K$. Then there exists an $i^* \in [\frac{1}{\eps}]$ such that at most $8\eps \cdot \mu(K)$ nodes in length-three augmenting paths are misclassified for $i^*$. 
\end{obs}

\begin{proof}
Since $M$ is a maximum matching in $K$, length-three augmenting paths contain at most $4\mu(K)$ nodes. Since Let $\{d^u(v)\mid u,v\in V\}$ are $(\eps^2d)$-approximate degrees, for a node $v$, we have that $\max_u \lvert d^u(v) - d_K(v)\rvert \leq \eps^2 d$. Hence, a node can be misclassified if the difference between its degree and the thresholds used to define a node's degree class is less than $\eps^2d$. Since the difference between thresholds in \Cref{def: nodes-class} for different $i$ and $j$ is at least $\eps^2d$, each such node can be misclassified for at most two such indices $i\in [\frac{1}{\eps}]$. Therefore, there exists an $i^*$ such that at most $8\mu(K)/(\frac{1}{\eps}) = 8\eps \cdot \mu(K)$ nodes are misclassified for $i^*$.
\end{proof}

We are now ready to generalize \Cref{lem: augmentingKernel} to the approximate degree setting.

\begin{lem}\label{lem: augmentingApproximateKernel}
    Let $i^*\in [\frac{1}{\eps}]$ be the index guaranteed by \Cref{lem: threshold}.
    Let $M_H$ and $M_{SH}$ be $(3/2 +  \eps)$-approximate matchings in $B_H^{(i^*)}$ and $B_{SH}^{(i^*)}$, respectively. Let $A(M_H)$ and $A(M_{SH})$ denote the number of disjoint augmenting paths induced by the union of $E_K$ and each of $M_H$ and $M_{SH}$, respectively. Then, 
    $$\max \left(A(M_H), A(M_{SH})\right) \geq (\delta + \eps)\cdot \mu(K).$$
\end{lem}
\begin{proof}
We slightly abuse notation, and say a length-three augmenting path is frequent if it is of types (1), (2), (3), or (4) as in \Cref{def: freq-aug} \textbf{and} none of its nodes are misclassified.
By \Cref{lem: manyNf} and \Cref{lem: threshold}, the number of frequent paths (under this new definition) is at least 
$n'_f \geq n_f - 8\eps\cdot \mu(K) \geq (1-100s-\frac{\delta}{s} - 8\eps)\cdot \mu(K)$.
Now that with this new definition, a node $v$ is \emph{bad} if $v$ belongs to both $V(M)$ and to an infrequent augmenting path, where infrequent paths also include paths with misclassified nodes. 
Therefore, by observations \ref{lem: fewBadNodes} and \ref{lem: threshold} the number of bad nodes (under this new definition) is at most $b' \leq b+ 16 \eps\cdot \mu(K) \leq (200s + \frac{2\delta}{s} + 16\eps)\cdot \mu(K)$.
The maximum matching size in $B_H$ and $B_{SH}$ are the same as their counterparts without misclassifications, only replacig $n_f$ and $n_i$ for $i\in [4]$ by their tagged counterparts. That is, we have $\mu(B_{SH}) \geq n'_f + n'_1$ and $\mu(B_{H})\geq n'_f + n'_1 + n'_2$.
The above bounds on the number of bad nodes, frequent paths, and maximum matching sizes in $B_H$ and $B_{SH}$ are precisely the loose bounds in the proof of the desired inequality in \Cref{lem: augmentingKernel}. 
Repeating the same final calculations as in that proof, this lemma follows.
\end{proof}

\Cref{lem: augmentingApproximateKernel} suggests a method for computing a sparse augmentation for $K$: compute large matchings in each $B^{(i)}_H$ and $B^{(i)}_{SH}$. Formally, we have the following.

\begin{thm}\label{core-approx-lem}
    For $i \in [\frac{1}{\eps}]$, let $M_H^{(i)}$ and $M_{SH}^{(i)}$ be a $(3/2 + \eps)$-approximate matching in $B_H^{(i)}$ and $B_{SH}^{(i)}$, respectively. The the edge set $A = \bigcup_{i \in [\frac{1}{\eps}]} \left(M_H^{(i)} \cup M_{SH}^{(i)} \right)$ augments $K$ as follows. $$\mu(G)\leq (2-\delta)\cdot \mu(G[E_K \cup A]).$$
\end{thm}

\begin{proof}
By \Cref{lem: augmentingApproximateKernel}, there exists an $i^* \in [\frac{1}{\eps}]$ such that $G[E_K\cup M_H^{(i^*)}\cup M_{SH}^{(i^*)}]$ has at least $(\delta + \eps)\cdot\mu(K)$ augmenting paths with regards to $M$. Therefore, since $\mu(G) \leq (2+8\eps)\cdot \mu(K)$, by \Cref{kernel:basic}, and since $\delta\leq 1$, we obtain the desired inequality.
\begin{align*}
    \mu(G[E_K \cup A]) & \geq \mu(G[E_K\cup M_H^{(i^*)}\cup M_{SH}^{(i^*)}]) \\ 
    & \geq \mu(K) + (\delta + \eps)\cdot \mu(K) \geq \frac{1 + \delta + \eps}{2+8\eps}\cdot \mu(G) \geqstar \frac{1}{2-\delta}\cdot \mu(G). \qedhere
\end{align*}
\end{proof}
In \Cref{sec:together} we show how to efficiently make use of the above structural characterizations of kernels and their augmentations to obtain the dynamic algorithm of our main result, \Cref{thm:det-beat-2}. But first, we require an efficient algorithm for maintaining kernels, which we present in the following section.

\section{Kernel Maintenance}\label{sec:kernel-maintenance}

In this section we describe a key component of our results: algorithms for kernel maintenance. 

\citet{bhattacharya2018deterministic}, in their paper introducing kernels, provided two algorithms for kernel maintenance, with amortized update times $O(\sqrt[3]{m}/\eps^2)$ and $O(\sqrt{n}/\eps)$, respectively.\footnote{Unlike in the preceding section, in this section $\eps>0$ will not be a fixed constant.}
They further showed how to de-amortize the former algorithm. Unfortunately for our needs, this algorithm maintains a kernel of maximum degree $d=\Omega(\sqrt[3]{m})$, and the number of changes to the kernel in this algorithm is $c=\Omega(\sqrt[3]{m})$. Consequently, maintaining near-maximum matchings in this kernel $K$ would require $c\cdot d = \Omega(m^{2/3})$ update time, which is super-linear in $n$ for moderately dense graphs.

In this section we show that the second algorithm of \cite{bhattacharya2018deterministic} can also be de-amortized (and sped up for sparse graphs). Crucially for our needs, the obtained algorithm will only cause $c=O(1)$ edge changes to the maintained kernel per edge update in $G$. This sub-linear update time and small number of edge changes in these kernels will prove useful for all of our results. We now turn to describing this kernel maintenance algorithm.

\subsection{The Basic Kernel Maintenance Algorithm}\label{sec:kernel-maintenance-algo-simple}

Next, we present a slightly simplified version of our kernel maintenance algorithm, with running time $O(n/(\eps d))$. 
For simplicity, we defer the details of our $O(\sqrt{m}/(\eps d))$ time algorithm to \Cref{sec:kernel-in-sqrtm-time}.

Our algorithm maintains the edge set $E_K$ using arrays and lists in a way to allow for insertion and deletions of edges in constant time. (See \cite{cormen2009introduction}.)
In addition, it maintains $E$ in a similar structure, with the refinement that for each vertex $v\in V$, its list of neighbors in $G$, denoted by $N_G(v)$, is stored in a bidirectional circular list, and we maintain a pointer $ptr(v)$ in this list.  

Our algorithm proceeds as follows: For insertion of an edge, we simply add it to $K$ if both endpoints have degree below $d$ in $K$ (i.e., insertion does not violate Property \ref{p1:bounded-deg}). Else, we do not change $K$. Regardless, for both endpoints of $e$, we insert the second endpoint into their linked list $N_G(v)$, \emph{before} the node $ptr(v)$ points to.
See the pseudocode in \Cref{alg:kernel-insert}.
Removal is the more intricate procedure. For this procedure, we maintain, for each vertex $v$, a pointer $ptr(v)$ to its list of edges in $G$, stored in circular order. When $v$ loses an edge in $K$ due to a removal of an edge in $G$, we move $ptr(v)$ through the next neighbors of $v$ in $G$, until we encounter a neighbor $w$ of $v$ in $G$ with $d_K(w)<d$. In this case, we add $(v,w)$ to $K$ and terminate the procedure. Otherwise, the procedure terminates after $n/(\eps d)$ neighbors of $v$ are inspected.
See the pseudocode in \Cref{alg:kernel-remove}.

    \begin{algorithm}[h] 
	\caption{\textsc{kernel-update}:\inserttwod$(e)$}
	\label{alg:kernel-insert}
	\begin{algorithmic}[1]
		\smallskip
		\If{$\max_{v\in e} d_K(V) < d$}
		\State $E_K\gets E_K\cup \{e\}$
		\EndIf
		\For{$v\in e$}
		\State insert $u\in e\setminus\{v\}$ into $N_G(v)$ before the element $ptr(v)$ is pointing to
		\EndFor
	\end{algorithmic}
\end{algorithm}

\begin{algorithm}[h] 
	\caption{\textsc{kernel-update}:\remove$(e)$}
	\label{alg:kernel-remove}
	\begin{algorithmic}[1]
		\smallskip
		\For{$v\in e$}
		\State remove $u\in e\setminus \{v\}$ from $N_G(v)$
		\EndFor
		\If{$e\in E_K$}
		\State{$E_K\gets E_K\setminus \{e\}$}
		\For{$v\in e$}
		\For{$i=1,2,\dots,n/(\eps d)$} \label{line:kernel-remove-loop-start}
		\State advance $ptr(v)$ in $N_G(v)$ \Comment{advance to next element in cyclic list of $N_G(v)$}
		\State $w\gets ptr(v).element$ \Comment{get node in $N_G(v)$ pointed to by $ptr(v)$}
		\If{$d_K(w)<d$}
		\State $E_K\gets E_K\cup \{(v,w)\}$\label{line:kernel-remove-fix}
		\State \Return \label{line:kernel-remove-loop-end}
		\EndIf
		\EndFor
		\EndFor
		\EndIf
	\end{algorithmic}
\end{algorithm}    

The intuition behind \Cref{alg:kernel-remove} is that for any edge not added to the kernel upon insertion, at least one of its endpoints $v$ must  have degree $d_K(v)=d$. 
The point of iterating through $n/(\eps d)$ neighbors of a vertex who lost an edge in $E_K$ is to guarantee that by the time $\eps d$ edges of this vertex have been deleted from $E_K$ without a replacement edge of $v$ being found in \Cref{line:kernel-remove-fix}, this vertex $v$ must have considered each of its $n$ (or fewer) neighbors $u$. When this happened, either $(u,v)$ was added to $E_K$ (and never removed), or $d_K(v)\geq d$. In either case, Property \ref{p2:satisfied-edges} is satisfied for edge $(u,v)$.

Formalizing the above discussion, we now analyze the above algorithm.

\begin{restatable}{lem}{kernelalgosimple}\label{kernel-maintenance-simple}
    For any $\eps>0$ and $d$, the algorithm handling edge insertions and deletions using algorithms \ref{alg:kernel-insert} and \ref{alg:kernel-remove} is a deterministic dynamic algorithm for maintaining an $(\eps,d)$-kernel using $O(n/(\eps d))$ worst-case update time and causing $O(1)$ changes to the kernel per update.
\end{restatable}
\begin{proof}
    The running time and number of changes to $E_K$ per update are immediate from the pseucodode. Property \ref{p1:bounded-deg} is also apparent, as we never add an edge to a vertex $v$ which already has $d_K(v)=d$. The only non-trivial property is Property \ref{p2:satisfied-edges}, which we now prove.
    
    Fix an edge $e\in E$ and time point $t$. 
    If $e$ was added to $E_K$ upon insertion of $e$, then, since edges are only removed from $E_K$ when they are removed from $G$, edge $e$ is still in $E_K$ at time $t$, and so it vacuously satisfies Property \ref{p2:satisfied-edges}. 
    Suppose instead $e$ was not added to $E_K$ upon insertion.
    Therefore, by \Cref{alg:kernel-insert}, at that time we had that $\max_{v\in e} d_K(v)=d$. Let $v\in e$ be such that $d_K(v)=d$ at the latest time $t'<t$. Now, if $d_K(v) < d(1-\eps)$ at time $t$, vertex $v$ must have had at least $\eps d$ edges deleted from $E_K$ without a replacement edge found in lines \ref{line:kernel-remove-loop-start}-\ref{line:kernel-remove-loop-end}. But then, during these $\eps d$ applications of the $n/(\eps d)$ iterations of lines \ref{line:kernel-remove-loop-start}-\ref{line:kernel-remove-loop-end} for vertex $v$, we  inspect all (at most) $n$ neighbors of $v$ which were not added after time $t'$, and in particular we have had in one of these iterations  $w=u$, where $u$ is the second endpoint of $e=(v,u)$. Since we chose $v$ to be the vertex in $e$ which had $d_K(v)$ at the last point $t'<t$, we have that $d_K(w)<d$, and so we add $e$ to $E_K$ in \Cref{line:kernel-remove-fix}. Consequently, again relying on edges only being removed from $E_K$ when they are removed from $G$, we have that edge $e$ is in $E_K$ at time $t$, and therefore vacuously satisfies Property \ref{p2:satisfied-edges} at time $t$.
\end{proof}

    In \Cref{sec:kernel-in-sqrtm-time}, we refine the result of \Cref{kernel-maintenance-simple}, by relying on the dynamic edge orientation algorithm of  \cite{bernstein2015fully}.
    This allows us to speed up our kernel algorithm for sparse graphs, replacing the $n$ factor in the running time of \Cref{kernel-maintenance-simple} by $\sqrt{m}\leq n$. From this, we obtain the following lemma.
\begin{restatable}{lem}{kernelalgo}\label{kernel-maintenance}
For any $\eps>0$ and $d$, there exists a deterministic dynamic algorithm for maintaining an $(\eps,d)$-kernel with $O(\sqrt{m}/(\eps d))$ worst-case update time, with $O(1)$ changes to the kernel per update.
\end{restatable}

\subsection{Applications: Fast $(2+\eps)$-Approximate (Weighted) Matching}
Since $(\eps,d)$-kernels have degree at most $d$ (by property \ref{p1:bounded-deg}), we can maintain near-maximum matchings in dynamic kernels in time $O_\eps(d)$ per edge change in kernels, by \Cref{bounded-deg-algo}. Combined with \Cref{kernel-maintenance} and \Cref{kernel:basic}, our deterministic kernel maintenance algorithms therefore yield the first deterministic $(2+\eps)$-approximate dynamic matching algorithm with worst-case sublinear update time. Specifically, we obtain the following.

\begin{thm}\label{thm:unweighted-kernel}
     There exists a deterministic $(2+\eps)$-approximate dynamic matching algorithm with worst-case update time $O(\sqrt[4]{m}/\sqrt{\eps}) = O(\sqrt{n}/\sqrt{\eps})$. 
\end{thm}
\begin{proof}
    Let $d=\sqrt[4]{m}/\sqrt{\eps}$.
    Our algorithm maintains an $(\eps,d)$-kernel $K$ in deterministic $O(\sqrt{m}/(\eps d))$ worst-case update time, \Cref{kernel-maintenance}.
    By \Cref{bounded-deg-algo} and property \ref{p1:bounded-deg}, we can maintain a $(1+\eps)$-approximate matching in $K$ in deterministic $O(d/\eps^2)$ worst-case update time per edge update in $K$. By \Cref{kernel-maintenance}, our kernel $K$ has $O(1)$ edges added or removed per update, and so maintaining this matching in $K$ contributes an $O(d/\eps^2)$ term to the (deterministic, worst-case) update time. Finally, by the kernel lemma \ref{kernel:basic}, this matching satisfies $|M|\geq \frac{1}{1+\eps}\cdot \mu(K)\geq \frac{1}{1+\eps}\cdot \frac{1}{2+\eps}\cdot\mu(G)$. Consequently, this algorithm maintains a $(1+\eps)\cdot (2+\eps) = 2+O(\eps)$ approximate matching in $G$ in deterministic $O(\sqrt{m}/(\eps d) + d/\eps^2)$ worst-case update time. By our choice of $d=\sqrt[4]{m}/\sqrt{\eps}$, the theorem follows.
\end{proof}

For weighted matching, we obtain the same asymptotic running time in terms of $n$ and $m$, albeit with an exponential dependence on $\eps$, by relying on the recent weighted matching framework of \citet{bernstein2021framework}. Letting $R$ denote the ratio of the maximum weight and minimum weights ever in the graph, we obtain the following. (See \Cref{sec:application:MWM} for details, as well as a discussion on exponentially decreasing this dependence on $R$.)

\begin{thm}\label{thm:weighted-kernel}
     There exists a deterministic $(2+\eps)$-approximate dynamic \emph{weighted} matching algorithm with worst-case update time $O(\sqrt[4]{m}/\eps^{O(1/\eps)}\cdot \log R) = O(\sqrt{n}/\eps^{O(1/\eps)}\cdot \log R)$.
\end{thm}

\Cref{thm:fast-kernel-algo} follows by combining the preceding two theorems.

\section{Putting it all Together: Beating the Folklore Algorithm}\label{sec:together}
In this section we combine our ideas from the previous section to obtain our main result---a deterministic algorithm which outperforms the folklore algorithm in terms of both approximation ratio and worst-case update time. Specifically, we prove the following.

\beatfoklore*

\begin{proof}
    As in \Cref{sec:kernels++}, we let $\eps=2\times 10^{-8}$ and $\delta=2\times 10^{-6}$.
    Set $d=m^{3/8}\geq 1/\eps$.
    This last inequality is WLOG, since if $m^{3/8} < 1/\eps$, then we can trivially maintain a $1$-approximate matching in worst-case time $\poly(n)=\poly(1/\eps)$ per update, by recomputing a maximum matching after each update.
    Furthermore, in this proof we assume that $m$ is fixed (up to some multiplicative constant). This too is WLOG, since whenever $m$ increases or decreases by a constant factor, we can (gradually) re-run the dynamic matching algorithm with an appropriately larger/smaller value of $m$, incurring a constant multiplicative overhead. 

    Our algorithm consists of a number of subroutines applied to several subgraphs of the dynamic graph $G$. Our algorithm maintains an $(\eps,d)$-kernel $K=(V,E_K)$. In addition, it maintains $\eps^2 d$-approximate degrees $\{d^{u}(v) \mid u,v\in V\}$ for all node pairs $u,v\in V$, and, for all $i\in [1/\eps]$, the subgraphs $B^{(i)}_{SH}$ and $B^{(i)}_H$ induced by these approximate degrees, as defined in \Cref{sec:main-approx-degrees}, where all (star) updates to approximate degree counters $d^{u}(v)$ affect the $O(1)$ appropriate subgraphs $B^{(i)}_{SH}$ and $B^{(i)}_{H}$.
    In each of the bipartite subgraphs $B^{(i)}_{SH}$ and $B^{(i)}_{H}$, we maintain $(3/2+\eps)$-approximate matchings, $M^{(i)}_{SH}$ and $M^{(i)}_H$, respectively, using the bipartite matching algorithm of \cite{bernstein2015fully}, making it incur only $O(1/\eps)$ recourse per star update in $B^{(i)}_{SH}$ and $B^{(i)}_H$, using \Cref{recourse}.
    Finally, we maintain a $(1+\delta/8)$-approximate matching to use as our output, $M^{out}$, in the low-degree graph $AK:=G\left[E_K\cup \bigcup_{i} \left(M^H_i \cup M^{SH}_i\right)\right]$, using \Cref{bounded-deg-algo}.
    
    As we shall show, the above is a  $((2-\delta)\cdot (1+\delta/8))$-approximate algorithm with $O(\sqrt[4]{m}\cdot\sqrt{\Delta})$ update time. To speed up this algorithm for sparse graphs, we combine our algorithm with \Cref{lem:arboricity} (essentially a corollary of the sparsification scheme of \cite{solomon2018local}), which allows us to decrease $\Delta$ to $O(\sqrt{m}/\eps)$ at the cost of increasing the approximation ratio by a $(1+\eps)$ factor. 
    

    
    \textbf{Approximation Ratio.} By \Cref{core-approx-lem}, we have that $\mu(G)\leq (2-\delta)\cdot \mu(AK)$. Also, by applying dynamic sparsification of \Cref{lem:arboricity} to our algorithm, we lose a $(1+\eps)$ factor in approximation ratio.
    Therefore, the $(1+\delta/8)$ approximate matching $M^{out}$ which we maintain satisfies the following.
    $$\mu(G)\leq (2-\delta)\cdot \mu(AK) \leq (2-\delta)\cdot (1+\eps)\cdot(1+\delta/8)\cdot |M^{out}| \leq (2-\delta/2)\cdot |M^{out}|.$$
    That is, this algorithm is $(2-10^{-6})$-approximate, as claimed.
    
    \textbf{Running Time.\footnote{All following bounds are worst-case, so we omit the words worst-case for brevity.}} 
    By \Cref{kernel-maintenance}, the maintenance of the $(\eps,d)$-kernel $K=(V,E_K)$ requires $O(\sqrt{m}/(\eps d))$ time and $O(1)$ changes to $E_K$ per update.
	For each insertion or removal of an edge $(u,v)$ in $E_K$, we update the approximate degrees, in $O(\Delta/(\eps^2 d))$ time, causing an insertion or removal of $O(1)$ stars (centered on $u$ or $v$) with $O(\Delta/(\eps^2 d))$ leaves in the bipartite subgraphs of the form $B^{(i)}_{SH}$ and $B^{(i)}_H$, by \Cref{alg:approx-degrees} and \Cref{B-maintenance}. Consequently, the time to maintain the auxiliary approximate matchings $\{M^{(i)}_{H}, M^{(i)}_{SH}\}_i$ in these bipartite subgraphs $B^{(i)}_{H}$ and $B^{(i)}_{SH}$ following a star update is $O(\Delta/(\eps^2 d))
    \cdot O(\sqrt[4]{m}/\poly(\eps)) = O(\Delta\cdot \sqrt[4]{m} /d)$, by \Cref{edcs-algo}. 
    Therefore, since only $O(1)$ such star updates occur per update in $G$, the time per update to maintain the $1/\eps=O(1)$ matchings $M^{(i)}_{SH}$ and $M^{(i)}_H$ is $O(\Delta\cdot \sqrt[4]{m} /d)$.
    Moreover, each star update results in $O(1/\eps)=O(1)$ changes to each of these auxiliary matchings, by \Cref{recourse}, and so $A:=\bigcup_i (M^{(i)}_{SH}\cup M^{(i)}_H)$ changes by $O(1/\eps)$ edges per update to $G$. We conclude that the graph $AK:=G[A\cup E_K]$ is a graph of maximum degree $d+O(1) = O(d)$, which changes by $O(1)$ edges per update: $O(1)$ changes to the kernel per update, by \Cref{kernel-maintenance-simple} and $O(1)$ changes in the matchings, by the preceding discussion.
    Therefore, we can maintain a $(1+\delta/8)$-approximate matching in $AK$ in $O(d/\delta^2)=O(d)$ time per update to $AK$, by \Cref{bounded-deg-algo}, which is $O(d)\cdot O(1)=O(d)$ time per update to $G$.
    We conclude that this algorithm's (worst-case) update time is $O(d+\Delta\cdot \sqrt[4]{m}/d)$, which can decrease to $O(d+m^{3/4}/d)$ for sparse graphs by using dynamic sparsification of \Cref{lem:arboricity}. Choosing $d=m^{3/8}$, yielding the claimed time bound.
    
    
\end{proof}

\paragraph{Acknowledgements.} The authors thank Sayan Bhattacharya and Aaron Bernstein for helpful discussions, and Aaron Bernstein for sharing a preprint of \cite{bernstein2021framework}.
This research is partially supported by NSF award 1812919, ONR award N000141912550, and a gift from Cisco Research.

\appendix
\section*{APPENDIX}
\section{Omitted Proofs of \Cref{sec:prelims}}\label{appendix:prelims}

Here we provide omitted proofs of some basic dynamic algorithmic primitives, apparent in prior work, which we use. These proofs will rely on the following simple observation, itself implicit in  \citet[Lemma 3.1]{gupta2013fully}, whereby the matching problem is in some sense ``stable'', and small changes do not affect the size of any matching much.

\begin{obs}\label{stability}
    Let $\alpha\geq 1$ and $\eps \leq 1/3$. If $M$ is an $\alpha$-approximate matching in $G_t$, then for all $t'\in [t, t+\eps\cdot|M|]$, the matching $M'$ consisting of the edges of $M$ which were not deleted during (edge/vertex/star) updates $t,t+1,\dots,t'$ is an $\alpha(1+3\eps)$-approximate matching in $G_{t'}$.
\end{obs}
\begin{proof}
    Each edge/vertex/star update can contribute/remove at most one edge of any matching. Therefore, during the $u = t'-t \leq \eps\cdot |M|$ updates following time $t$, both the maximum matching size and the size of the non-deleted subset of $M$ cannot change by more than $\eps\cdot |M|\leq \eps\cdot \mu(G_t)$. 
    Consequently, we have that 
    \begin{align*}
        \mu(G_{t'}) & \leq \mu(G_t) + u \leq (1+\eps)\cdot \mu(G_t) \\
        |M'| & \geq |M|-u \geq \left(1-\eps\right)\cdot |M|.
    \end{align*}
    Combining both bounds, the claimed bound follows by $M$ being $\alpha$-approximate in $G_t$ implying $\mu(G_t)\leq \alpha\cdot |M|$, together with our choice of $\eps\leq 1/3$, which imply
    \begin{align*}
        \mu(G_{t'}) & \leq (1+\eps)\cdot \mu(G_t) \leq \alpha(1+\eps) \cdot |M| \leq \frac{\alpha (1+\eps)}{1-\eps}\cdot |M'| \leq \alpha\cdot (1+3 \eps)\cdot |M'|. \qedhere
    \end{align*}
\end{proof}

\Cref{stability} plays a crucial role in a number of prior results for the dynamic matching problem, many of which rely on fast algorithms for bounded-degree graphs, as in the following proposition, implied by \cite{gupta2013fully,peleg2016dynamic}. 

\boundeddeg*
\begin{proof}[Proof (Sketch)]
    Let $\eps'=\eps/5$. The algorithm divides the update sequence into phases, with the phase starting with matching $M$ consisting of $\lfloor \eps'\cdot |M|\rfloor+1$ updates.
    At the beginning of each phase (i.e., after its first update), we run a static $(1+\eps')$-approximate matching algorithm in $G$ in time $O(|E|/\eps')$ (e.g., \cite{micali1980v}).
    Since nodes of a maximum matching form a vertex cover with $2\mu(G)$ nodes, $G$ has at most $2\mu(G)\cdot \Delta$ edges, and so this static algorithm takes $O(|E|/\eps') = O(\mu(G)\cdot \Delta/\eps)$ time. 
    During the phase starting with matching $M$, we use the edges of $M$ not deleted during the phase as our matching.
    By \Cref{stability}, this results in approximation ratio at most $(1+\eps')\cdot (1+3\eps')\leq (1+5\eps') = (1+\eps)$.
    The amortized update time of this algorithm is $O(\Delta/\eps^2)$. This is easily de-amortized by spreading out the work of the static $(1+\eps')$-approximate matching over subsequent updates, and again appealing to \Cref{stability} to argue that this does not deteriorate the approximation ratio by more than a further $(1+O(\eps'))$ multiplicative factor, resulting in a $1+\eps$ approximation, by appropriately re-parameterizing.
\end{proof}

Next, we provide a brief proof sketch of \Cref{recourse}, originally stated for edge or vertex updates by \citet{solomon2021generalized}, and show that their argument carries through seamlessly to star updates as well.
\recourse*
\begin{proof}[Proof (Sketch)]
    In \cite[Theroem 1]{solomon2021generalized}, Solomon and Solomon provide an algorithm for transforming one matching $M$ into (a possible super-set of) another matching $M'$ in time $O(|M|+|M'|)$, such that each intermediate matching in this transformation has size at least $\min\{|M|,|M'|-1\}$, and every two consecutive matchings differ by $O(1)$ edges. 
    To obtain a bounded-recourse dynamic algorithm from the above, one again divides the update sequence into phases, with a phase starting with a matching $M$ consisting of the following $\lfloor \eps\cdot |M|\rfloor$ updates. Let $M_i$ be the matching maintained by the dynamic matching algorithm at the end of phase $i$. 
    At the end of phase $i$ (i.e., beginning of phase $i+1$), we gradually transform $M_i$ into $M_{i+1}$, by performing $\Theta(1/\eps)$ operations of the operations in the transformation of \citet[Theorem 1]{solomon2021generalized}, such that at the end of the phase, all $O(|M_i|+|M_{i+1}|)$ operations were performed.
    This trivially contributes $O(1/\eps)$ to the update time and recourse per update.
    Finally, by \Cref{stability}, the matching maintained this way, which has size at least $\min\{|M_i|,|M_{i+1}|-1\}$, is an $\alpha(1+6\eps)$-approximate matching, unless $|M_{i+1}|=1$, in which case we can maintain a $1\leq \alpha$ approximation by maintaining a matching consisting of a single edge.
\end{proof}

\section{Omitted Proofs of \Cref{sec:kernels++}}\label{appendix:kernel++}

Here we provide full proofs of results of \Cref{sec:kernels++}. We use the same settings of $\eps,\delta,s$ and $d$ as in that section, and again use $\leqstar$ and $\geqstar$ to denote inequalities which hold for our particular choices of these variables.

\subsection{The Extended Kernel Lemma}
In this section we revisit the proof that kernels contain a large matching compared to their supergraph's maximum matching \cite{arar2018dynamic,bhattacharya2018deterministic,wajc2020rounding}. Specifically, as stated in \Cref{kernel:basic}, a kernel is $(2+\eps)$-approximate matching sprasifier. We extend this lemma to prove that if a kernel is not much better than $2$-approximate, then this imposes restrictions on degress in the kernel of endpoints of edges in a maximum matching $M^*$ in $G$. Specifically, we prove the following.
\ekl*
\begin{proof}
    To prove that $K$ contains a large matching, we first prove that it contains a large fractional matching. 
    We consider the following such fractional matching in $K$, using the shorthand notation $z^+ := \max\{0,z\}$.
\begin{equation}\label{fractional-matching}
    x_e = \begin{cases} \frac{1}{d} & e\in E_K\setminus M^* \\
    \left(1-\sum_{v\in e}\frac{d_K(v)-1}{d}\right)^+ & e\in E_K\cap M^*,
    \end{cases}
\end{equation}
It is easy to see that $\vec{x}$ is a fractional matching in $K$, i.e., $x_e\geq 0$ and $\sum_{e\ni v} x_e\leq 1$.
We will now prove that this fractional matching's size is large compared to $|M^*|=\mu(G)$, i.e., that is satisfies $\sum_e x_e \geq \frac{1}{\alpha}\cdot |M^*|$ for some small $\alpha\leq 2+O(\eps)$, and that moreover, $\alpha\leq 2-\Omega(\delta)$ if many edges of $M^*$ satisfy either condition \eqref{high-degree-sum-condition} or \eqref{low-degree-sum-condition}.
By the integraligy of the fractional matching polytope in bipartite graphs, such a bound on $\sum_e x_e$ immediately implies that if $G$ is bipartite, then the largest \emph{integral} matching in $K$ has size at least $\mu(K)\geq \sum_e x_e \geq \frac{1}{\alpha}\cdot |M^*|$. 
In general graphs, where sizes of fractional matchings may be a factor of $\frac{3}{2}$ greater than the cardinality of the maximum (integral) matching, $\sum_e x_e \geq \frac{1}{\alpha}\cdot \mu(G)$ would na\"ively imply a $\frac{3\alpha}{2}$ approximation. 
However, prior work \cite{arar2018dynamic,bhattacharya2018deterministic,wajc2020rounding} shows that for this particular fractional matching $\vec{x}$, this factor of $\frac{3}{2}$ blowup is not necessary. In particular, they show that there exists an integral matching in $K$ of cardinality nearly $\sum_e x_e$. More precisely, relying on Vizing's theorem \cite{vizing1964estimate}, they show the following. (The proof is provided below for completeness.)
\begin{restatable}{lem}{vizingapp}\label{vizing-app}
     Any $(\eps,d)$-kernel $K=(V,E_K)$ and  fractional matching $\vec{x}$ as in \eqref{fractional-matching} satisfy $\mu(K)\geq \frac{\sum_e x_e}{1+1/d}.$
\end{restatable}

By \Cref{vizing-app}, to prove the basic kernel lemma, \Cref{kernel:basic}, it is sufficient to prove that $\sum_e x_e \geq \frac{1-\eps}{2}\cdot |M^*|$. 
For this, one can (as we shall do) prove that for every edge $e\in M^*$, we have that $\sum_{e'\cap e} x_{e'}\geq 1-\eps$.
In what follows, we use a refinement of this bound to prove that if \Cref{kernel:basic} is not too loose, i.e., $\mu(G)\geq (2-\delta)\cdot \mu(K)$, then most edges $e\in M^*$ must have that their endpoints' degrees in $K$ have a bounded sum, as asserted in the lemma statement.

    Let $y_v:=\sum_{e\ni v} x_e$. 
    We obtain a lower bound on $\sum_e x_e = \frac{1}{2}\cdot \sum_v y_v$ by lower bounding $\sum_{v\in V(M^*)} y_v$ in terms of $|M^*|$, by lower bounding $\sum_{v\in e} y_v$ for each edge $e\in M^*$.
    
    If $e\in M^*\setminus E_K$, then $\max_{v\in e} d_K(v)\geq d(1-\eps)$, by property \ref{p2:satisfied-edges} of the $(\eps,d)$-kernel $K$, and so
    \begin{align}\label{out-of-kernel-dual-bound}
    \sum_{v\in e} y_v = \sum_{v\in e} d_K(v)\cdot \frac{1}{d} \geq 1-\eps.
    \end{align}
    Conversely, if $e\in M^*\cap E_K$, we have that \begin{align}\label{in-kernel-dual-bound} 
    \sum_{v\in e} y_v = \sum_{v\in e}\frac{d_K(v)-1}{d} + 2\cdot \left(1-\sum_{v\in e}\frac{ d_K(v)-1}{d}\right)^+ \geq 1 \geq 1-\eps.
    \end{align}
    
    Next, we revisit the bounds of equations \eqref{out-of-kernel-dual-bound} and \eqref{in-kernel-dual-bound}, 
    and bound $\sum_{v\in e} y_v$ for edges $e\in M^*$ satsifying either of conditions \ref{high-degree-sum-condition} or \ref{low-degree-sum-condition}.
    For edges $e\in M^*$ satisfying condition \ref{high-degree-sum-condition}, the endpoints' degree sum in $K$ satisfies $\sum_{v\in e} d_K(v) \geq  (1+s-2\eps)\cdot d$, and so $\sum_{v\in e} (d_K(v)-1) \geq d\cdot (1+s-4\eps)$, since $d\geq 1/\eps$. The bounds of equations \eqref{out-of-kernel-dual-bound} and \eqref{in-kernel-dual-bound} then imply the following.
    \begin{align}\label{high-degs-dual-bound}
        \sum_{v\in e} y_v \geq \sum_{v\in e}\frac{d_K(v)-1}{d} \geq 1+s-4\eps.
    \end{align}
    Similarly, edges $e\in M^*$ satisfying condition \ref{low-degree-sum-condition} satisfy $\sum_{v\in e} d_K(v)\leq d\cdot (1-s+2\eps)$ and $e\in E_K$, and hence, 
    \begin{align}\label{low-degs-dual-bound} 
    \sum_{v\in e} y_v = \sum_{v\in e}\frac{d_K(v)-1}{d} + 2\cdot \left(1-\sum_{v\in e}\frac{ d_K(v)-1}{d}\right)^+ = 1 + \left(1-\sum_{v\in e}\frac{ d_K(v)-1}{d}\right) \geq 1+s-2\eps.
    \end{align}
    
    Denote by $r$ the number of edges $e$ in $M^*$ which satisfy either condition \ref{high-degree-sum-condition} or \ref{low-degree-sum-condition}. Equations \eqref{out-of-kernel-dual-bound}
    --\eqref{low-degs-dual-bound} then imply the following lower bound on the size of $x$.
    \begin{align}\label{12345}
        \sum_e x_e & = \frac{1}{2}\cdot \sum_v y_v \geq \frac{1}{2}\cdot \sum_{v\in V(M^*)} y_v \geq \frac{1}{2}\cdot \left(|M^*|\cdot (1-\eps) + r\cdot (s - 3\eps)\right).
    \end{align}
    On the other hand, we have by \Cref{vizing-app} and $d\geq 1/\eps$ that 
    \begin{align}\label{123456}
        \mu(K)\geq \frac{1}{1+1/d}\cdot \sum_e x_e \geq \frac{1}{1+\eps}\cdot \sum_e x_e.
    \end{align}
    Now, by the hypothesis of the lemma, we have that $\mu(G)\geq (2-\delta)\cdot \mu(K)$. 
    Combining this with equations \eqref{12345} and \eqref{123456}, we therefore have that 
    \begin{align*}
        \frac{1}{2-\delta} \cdot \mu(G)\geq \mu(K) \geq \frac{1}{1+\eps}\cdot x_e \geq \frac{1}{1+\eps} \cdot \frac{1}{2}\cdot \left(\mu(G)\cdot (1-\eps) + r\cdot (s - 3\eps)\right).
    \end{align*}
    Rearranging terms, we find that the number $k$ of edges of $M^*$ satisfying either condition \eqref{high-degree-sum-condition} or \eqref{low-degree-sum-condition} is indeed bounded as claimed, due to our choice of $\eps,\delta$ and $s$.
    \begin{align*}
        r & \leq \frac{2(1+\eps)}{s-3\eps}\cdot \left(\frac{1}{2-\delta} - \frac{1-\eps}{2(1+\eps)}\right) \cdot \mu(G) \\
        & \leqstar \frac{2(1+\eps)}{s-3\eps}\cdot \left(\frac{1}{2}+\delta/3-\left(\frac{1}{2}-\eps\right)\right) \cdot \mu(G) \\
        & = \frac{2(1+\eps)(\delta/3+\eps)}{s-3\eps} \cdot \mu(G). \\
        & \leqstar \frac{\delta}{s}\cdot \mu(G). \qedhere
    \end{align*}
\end{proof}

Finally, we provide a brief proof of \Cref{vizing-app}.
\begin{proof}[Proof of \Cref{vizing-app}]
Consider the graph $H:=G[E_K\setminus M^*]$. This subgraph of $K$ has maximum degree at most $d$, and hence is $(d+1)$-edge-colorable, by Vizing's theorem \cite{vizing1964estimate}. 
Now, consider the multigraph obtained by adding, for each edge $e\in E_K\cap M^*$, some $x_e\cdot d = (d-\sum_{v\in e}(d_K(v)-1))^+$ parallel copies of $e$ to $H$, yielding a multigraph $\mathcal{H}$.
For any edge $e\in E_K\cap M^*$, the $(d+1)$-edge-coloring of $H$ uses at most $\sum_{v\in e}(d_K(v)-1)$ distinct colors on the edges sharing an endpoint with $e$. Consequently, there are at least $(d+1 - \sum_{v\in e}(d_K(v)-1))^+\geq (d- \sum_{v\in e}(d_K(v)-1))^+$ unused colors among the palette of size $d+1$ used to color $H$, and so this multigraph $\mathcal{H}$ is also $(d+1)$-edge-colorable. 

Now, we note that for every edge $e\in H$ we have $x_e\cdot d$ parallel copies of $e$ in the multigraph $\mathcal{H}$. 
Consequently, this graph has $\sum_e x_e\cdot d$ many edges (counting multiplicities). 
Therefore, by averaging over the $d+1$ matchings (colors) in the $(d+1)$-edge-coloring of $\mathcal{H}$, we find that this multigraph (and therefore its simple graph counterpart, $K$) has a large matching in terms of $\vec{x}$, as claimed.
\begin{align*}
    \mu(K) & = \mu(\mathcal{H}) \geq \frac{\sum_e x_e\cdot d}{d+1} = \frac{1}{1+1/d}\cdot \sum_e x_e. \qedhere
\end{align*}
\end{proof}

\subsection{Characterizing augmenting path node classes}\label{appendix:characterizing}

\Hfree*
\begin{proof}
First, by Vizing's theorem \cite{vizing1964estimate}, the kernel $K$, which has maximum degree $d$, can be $(d+1)$-edge-colored. Each of these colors is a matching in $K$. Let $M'$ be a color (matching) uniformly sampled from these colors. Let $U$ be the set of nodes in $V_H$ that are not matched in $M'$. Since $M'$ is a randomly chosen color, a node $v \in V_H$ is matched with probability at least $\frac{d\cdot (1-2s-i\eps^2)}{d+1} \geq \frac{d\cdot(1-2s-\eps)}{d+1} \geqstar 1 - 2.1s$. 
Hence, by linearity of expectation, we have that $\mathbb{E}[|U|] \leq 2.1s \cdot |V_H|.$

Next, we upper bound $|V_H|$ in terms of $\mu(K)$, relying on the $\frac{3}{2}$ integrality gap of the fractional matching polytope. 
Consider a fractional matching such that each edge $e$ has a value equal to $x_e=1/d$. Since each node $v$ in $V_H$ has degree at least $d_K(v)\geq d\cdot (1-2s-i\cdot \eps^2)\geq d\cdot (1-2s-\eps)$, we have that
\begin{align*}
    \frac{3}{2}\cdot \mu(K) \geq \sum_e x_e \geq \frac{|V_H|}{2}\cdot\frac{1}{d}\cdot d\cdot(1-2s-\eps) = \left(\frac{1}{2} - s - \frac{\eps}{2}\right)\cdot|V_H| \geqstar \frac{6}{13}\cdot |V_H|.
\end{align*}
Dividing through by $\frac{6}{13}$, we obtain $|V_H|\leq \frac{13}{4}\cdot \mu(K) = 3.25\cdot \mu(K).$ Combining this bound with our upper bound on $\mathbb{E}[|U|]$, we find that
\begin{align*}
\mathbb{E}[|U|] \leq 2.1s \cdot |V_H| \leq 7s\cdot \mu(K).
\end{align*}

By the above, there exists some (not necessarily maximum) matching $M'$ in $K$ leaving at most $7s\cdot \mu(K)$ nodes in $V_H$ unmatched.
Let $\tilde{M}$ be a maximum matching in $K$. Note that $\tilde{M} \bigoplus M'$ consists of disjoint cycles and paths and exactly $|\tilde{M}| - |M'|$ augmenting paths. 
Let $P$ be this set of augmenting paths. Then, $M:=M'\bigoplus P$ is a maximum matching of size $|\tilde{M}|=\mu(K)$ matching all nodes matched by $V(M')$ (and possibly others). Therefore, the obtained matching $M$ is a maximum matching in $K$ which leaves at most $7s\cdot \mu(K)$ high degree nodes unmatched.
\end{proof}

\manyaug*
\begin{proof}
Let $U$ be the set of nodes in $V_H$ that are not matched by $M$. First, we claim that $M$ is a maximal matching in $G[V\setminus U]$. First, since $M$ is a maximum matching in $K$, there is no edge in $E_K$ with both of its endpoints unmatched by $M$. On the other hand, by Property \ref{p2:satisfied-edges} of the kernel, for any edge $e \notin E_K$, we have $\max_{v\in e} d_{H}(v) \geq d(1-\eps)$. Hence, at least endpoint should be in $V_{SH}\subseteq V_H$, but all $V_H$ nodes in $G[V\setminus U]$ are matched in $M$. We conclude that $M$ is indeed maximal.

Now, $M$ is a maximal matching in $G[V\setminus U]$, whose size satisfies
\begin{align*}
    |M| = \mu(K) \leq \frac{1}{2-\delta}\cdot \mu(G) \leq \frac{1}{2-\delta}\cdot \mu(G[V\setminus U]) + 7s\cdot \mu(K), 
\end{align*}
where the last inequality follows from \Cref{lem: Hnodes-matches}.
Re-arranging terms, we get that 
\begin{align*}
|M|\leq \frac{1}{(2-\delta)\cdot (1-7s)}\cdot \mu(G[V\setminus U]) \leqstar \left(\frac{1}{2}+2\delta+14s\right)\cdot \mu(G[V\setminus U]).
\end{align*}
Consequently, combining \Cref{lem: Hnodes-matches}, the assumption that $\mu(G)\leq (2-\delta)\cdot \mu(K)$, and \Cref{pre: threeAugmenting}, we find that the number of 3-augmentable edges of $M$ is at least 
\begin{align*}
    \left(\frac{1}{2} - 6\delta - 42s \right)\cdot \mu(G[V\setminus U]) & \geq \left(\frac{1}{2} - 6\delta - 42s\right)\cdot(\mu(G) - 7s\cdot \mu(K)) \\
    & \geq \left(\frac{1}{2} - 6\delta - 42s\right)\cdot (2-\delta)\cdot \mu(K) - \frac{1}{2}\cdot 7s\cdot \mu(K) \\
    & \geq \left(1 - \frac{25}{2}\delta - 84s - \frac{7}{2}s \right) \cdot \mu(K) \\
    & \geqstar \left(1- 88s \right)\cdot\mu(K). \qedhere
\end{align*}
\end{proof}

\subsection{Maintaining Approximate Degrees in Dynamic Subgraphs}\label{sec: approx-deg-maintenance}

In this section, we give a method for dynamically maintain $\alpha$-approximate degrees $\{d^u(v)\mid u,v\in V\}$ of a dynamic subgraph $K$, as well as bipartite subgraphs $B$ determined by edges $(u,v)$ with $d^u(v)\geq d_1$ and $d^v(u)\leq d_2$, for some $d_1$ and $d_2$.\footnote{We note that subgraphs determined by more elaborate conditions of these approximate degree conditions can similarly be maintained, though this level of generality will not be necessary for our use.}
Crucially for our final dynamic matching algorithm, the additive error of $\alpha$ in these approximate degrees will allow $G'$ to change by $O(1)$ \emph{star updates} with $O(\Delta/\alpha)$ leaves for every change to $K$.

\begin{lem}\label{alg:approx-degrees}
     For $\alpha>0$, there exists a deterministic algorithm which, for a dynamic subgraph $K=(V,E_K)$ of $G=(V,E)$, where $E_K$ changes by $O(1)$ edges per update in $G$, maintains in $O(\Delta/\alpha)$ worst-case update time $\alpha$-approximate degrees $\{d^u(v)\mid u,v\in V\}$ in $K$.
     Moreover, for any update (insertion or deletion) of edge $(u,v)$ in $E_K$, only $O(\Delta/\alpha)$ counters of the form $d^w(u)$ and $d^w(v)$ are changed.
\end{lem}
\begin{proof}
Our algorithm for maintaining $\alpha$-approximate degree counters associates nodes in $V$ with $[n]$ and stores, for each node $v$, a pointer $p_v$, initially points to the first neighbor of the the $v$ in its adjacency list. 
Whenever a node $v$ has its degree change in $K$, we simply move $p_v$ to the next element in adjacency list (setting $p_v$ to points to the first element if we reach to the end of the list) a number of times, specifically, $\lceil \Delta/\alpha\rceil$ times. When we advance $p_v$ and it points to $u$, we set $d^u(v) \gets d_K(v)$. This clearly takes $O(\Delta/\alpha)$ time. Also, after an edge $(u, v)$ is inserted to the graph, we insert $u$ exactly before $p_v$ in the adjacency list of $v$ and set $d^{u}(v) = d_K(v)$ (similarly, we update the adjacency list of $u$). Moreover, every insertion/deletion of an edge $(u,v)$ into/from $E_K$ results in $O(\Delta/\alpha)$ updates to approximate degrees, all of the form $d^{w}(u)$ and $d^w(v)$, resulting in the desired star updates.
Finally, since $p_v$ points to $u$ (and hence $d^u(v)\gets d_K(v)$) at some point during any sequence of $\alpha$ edge updates to $E_K$ affecting the degree of $v$ in $K$, we have that $d^u(v) = d_K(v) \pm \alpha$.
\end{proof}

The bipartite subgraph $B=(V,E_B)$ with $E_B=\{(u,v)\in E \mid d^u(v)\geq d_1, \, d^v(u)\leq d_2\}$ as required by our main algorithm can be maintained trivially in the same asymptotic time as needed to maintain approximate degrees, by adding/removing edges $(u,v)\in E$ to/from $E_B$ whenever the appropriate approximate degree conditions are met/violated. Hence, we have the following.
\begin{cor}\label{B-maintenance}
    There exists a deterministic algorithm for maintaining subgraphs $B^H_i$ and $B^{SH}_i$ as in \Cref{def: threshold}, with each change to $K$ resulting in $O(1)$ star updates with $O(\Delta/(\eps^2 d))$ edges.
\end{cor}

\section{Omitted Proofs of \Cref{sec:kernel-maintenance}: Speeding up Kernel Maintenance}\label{sec:kernel-in-sqrtm-time}
In \Cref{sec:kernel-maintenance} we gave an $O(n/(\eps d))$ update time algorithm for kernel maintenance. In this section we refine this bound, improving it for sparse graphs, by replacing the $n$ factor by $\sqrt{m}\leq n$. In particular, we prove \Cref{kernel-maintenance}, restated below for ease of reference.
\kernelalgo*
For this proof, we will rely on dynamic edge orientations. An \emph{edge orientation} of a graph is an assignment, for each edge $(u,v)$, of a direction, either from $u$ to $v$, or vice versa.
An orientation is an $\alpha$-orientation if each node $v$ has at most $\alpha$ edges oriented away from it.
Every graph admits an $O(\sqrt{m})$-orientation. 
For our needs, we rely on a constant-time dynamic edge orientation algorithm, due to  \cite{bernstein2015fully}, which matches this bound using optimal time and edge flips (i.e., re-orientations).
\begin{lem}[\cite{bernstein2015fully}]
     There exists a deterministic algorithm for maintaining an $O(\sqrt{m})$-orientation with $O(1)$ worst-case update time and $O(1)$ flips per edge update in the worst case.
\end{lem}

We now turn to presenting an algorithm with properties asserted by \Cref{kernel-maintenance}. 

\subsection{The Algorithm}\label{sec:kernel-maintenance-algo}

For simplicity, we first discuss our algorithm assuming a fixed edge orientation. The small number of flips per edge update will allow us to extend this result directly to the more realistic setting that edge update to $G$ require changing the edge orientation, at the end of this section.

\medskip\noindent\textbf{High-level intuition.} Our algorithm of \Cref{kernel-maintenance} will follow an approach similar to the algorithm of \Cref{kernel-maintenance-simple}. 
The crucial difference is that now, when an edge $(u,v)\in E_K$ is removed from $G$, instead of having $u$ and $v$ scan through $O(n/(\eps \cdot d))$ neighbors, we will have them scan through only $O(\sqrt{m}/(\eps \cdot d))$ out-neighbors (according to the edge orientation). This is, however, insufficient to satisfy Property \ref{p2:satisfied-edges} on its own. To see this, note that some in-neighbors $w$ of $u$, for example, may have $K$-degree strictly less than $d(1-\eps)$ at some point, after which $u$ may have its $K$-degree likewise decrease below $d(1-\eps)$, without ever inspecting its incoming edges in the interim. Consequently, an edge $(u,w)$ as above would violate Property \ref{p2:satisfied-edges}. To address this problem, we let each node $u$ have a \emph{pool} of in-neighbors with degree strictly less than $d$. When $u$ has its $K$-degree decrease, it first checks, in $O(1)$, whether its pool contains some in-neighbor $w$. If so, we add $(u,w)$ to $E_K$. This resolves all possible violations of Property \ref{p2:satisfied-edges} including $u$, since its degree remains unchanged.
This raises another problem, however, since now $w$ has its degree increase. If this happens often enough, then it belongs to pools of out-neighbors $v$, who may pick $w$ in their pool when their $K$-degree decreases. But if $w$ has its $K$-degree decrease, then the edge $(v,w)$ cannot be added to the kernel $K$, as its addition would violate Property \ref{p1:bounded-deg}. Put otherwise, $w$ should not be in the pool of any out-going neighbor $v$ if $w$ reaches degree $d$. This can be achieved by spending $O(\sqrt{m}/(\eps d))$ per edge addition of $w$ to $K$, by maintaining the following invariant.

\medskip\noindent\textbf{Invariant: pools.} Each vertex $v$ has a pool $p_v$ stored as a doubly-linked list containing a subset of its in-neighbors $w$ with $d_K(w)<d$. Each vertex $w$ with out-going degree $d_{out}$ in the orientation belongs to the pools of $\min\{d_{out}, (d-d_K(w))\cdot \frac{\sqrt{m}}{\eps d}\}$ of its out-going neighbors. 

\smallskip 

The above invariant, together with the properties \ref{p1:bounded-deg} and \ref{p2:satisfied-edges} of the kernel $K$, are easy to maintain in $O(\sqrt{m}/(\eps d))$ using the edge (static) orientation, as follows.
In addition to the above, for each node $v$, we keep two linked lists, $L_{in}$ and $L_{out}$, whose union contains all of $v$'s outgoing neighbors in the orientation. List $L_{in}$ contains all out-neighbors $w$ of $v$ such that $v\in p_w$, as well as a pointer to the relevant element in $p_w$ (allowing to remove $v$ from $p_w$ in $O(1)$ time when accessing $w$ in the first list). The second list, $L_{out}$ contains all out-neighbors $w$ of $v$ such that $v\notin p_w$, as well as a pointer to the pool $p_w$ (allowing to add $v$ to $p_w$ in $O(1)$ time, when accessing $w$ in the second list.)
Using these lists, whenever the $K$-degree of $v$ increases or decreases, we can fix the pool invariant by scanning through $\sqrt{m}/(\eps d)$ elements of either $L_{out}$ or $L_{in}$ in a fairly simple manner, as we now describe.

\medskip\noindent\textbf{The algorithm.} Our algorithm works as follows: when edge $e$ is added, if both endpoints have degree strictly less than $d$, we add $e$ to $E_K$. Otherwise, if the pool invariant is violated (i.e., $e=(u,v)$ is oriented from $u$ to $v$, and node $u$ now belongs to strictly less than $\{d_{out},(d-d_K(u))\cdot \sqrt{m}{\eps d}\}$ out-neighbors), we add $u$ to $p_v$.
Edge deletions are the more intricate subroutine. When an edge $e\notin E_K$ is removed, then if the pool invariant is violated and $e=(u,v)$ is oriented from $u$ to $v$, then we add $u$ to $p_v$. This concludes the operation due to a removal of an edge $e\not\in E_K$.
When an edge $e\in E_K$ is removed, we do the following for both endpoints $v\in e$: first, check if $p_v$ is not empty. If so, pick some $w\in p_v$ and add $(v,w)$ to $E_K$. Then, fix the pool invariant for $w$, by removing $w$ from the list of some $O(\sqrt{m}/(\eps d))$ out-neighbors of $w$.
If $p_v$ is empty, fix the pool invariant for $v$ by scanning through $\sqrt{m}/(\eps d)$ out-going edges $(v,z)$ of $v$ for which $v\not\in p_z$. If any of these neighbors $z$ has $d_K(z)<d$, add $(v,z)$ to $E_K$. Otherwise, add $v$ to the pool $p_z$ for all the $\sqrt{m}/(\eps d)$ out-neighbors scanned.

\medskip\noindent\textbf{Dealing with edge flips.} Edge updates can also result in edges being flipped in the (dynamic) edge orientation. To address this, whenever an edge $e$ is flipped, we process this by removing and adding it back into the graph with the appropriate orientation.

\subsection{Analysis}

We now analyze the above algorithm, and prove that it is precisely a kernel-maintenance algorithm with the desired time and change bounds.
\begin{proof}[Proof of \Cref{kernel-maintenance}]
    We analyze the algorithm described in \Cref{sec:kernel-maintenance-algo}.
    The claimed (worst-case) running time and number of changes per update are immediate from the algorithm's description. 
    We note moreover that the algorithm explicitly fixes any violation of the pool invariant throughout its execution. We now prove by induction on the number of updates that this invariant implies the two kernel properties of the maintained kernel, $K$, which trivially hold on initialization, when the graph $G$ and kernel $K$ are empty.
    
    The maintenance of Property \ref{p1:bounded-deg} of the kernel follows from our never adding an inserted edge $e$ to $E_K$ when $\max_{v\in e} d_K(v)=d$. Indeed, the only time we add an edge $(u,v)$ to $E_K$ without explicitly checking this condition is when $u\in p_v$ is oriented from to $v$ and the $K$-degree of $v$ decreases. But in that case, since Property \ref{p1:bounded-deg} held up to the preceding time step, we have that the $K$-degree decrease of $v$ resulted in $d_K(v)<d$, and since the pool invariant holds throughout, we have that $u\in p_v$ implies that we also have $d_K(u)<d$. Therefore, no addition of an edge to $E_K$ violates Property \ref{p1:bounded-deg} (and trivially, no removal of an edge may violate this property).
    
    We now prove that this algorithm satisfies Property \ref{p2:satisfied-edges}. Consider an edge $(u,v)$, oriented from $u$ to $v$ at time $t$.
    If $d_K(u)\geq d(1-\eps)$, then the condition of Property \ref{p2:satisfied-edges} is satisfied for this edge.
    Otherwise, we have by the pool invariant that $u\in p_v$, since $u$ belongs to the pools of all of its outgoing edges in this case. 
    Consider the last time $t'<t$ when $u$ was added to $p_v$.
    We cannot have that $d_K(v)<d$ at time $t'$, as this would result in $(u,v)$ being added to $E_K$, rather than $u$ being added to $p_v$.
    Therefore, by Property \ref{p1:bounded-deg}, we have that $d_K(v)=d$ at time $t'$. 
    Consequently, after each update which decreases $d_K(v)$, we have that $p_v\ni u$, and so $p_v$ is not empty, implying that for some in-neighbor $w$ of $v$, we add the edge $(v,w)$ to $E_K$, thus restoring $d_K(v)$ to its previous value of $d$. We conclude that the condition of Property \ref{p2:satisfied-edges} is satisfied for edge $(u,v)$ at time $t$.
\end{proof}

\section{Application: Dynamic Weighted Matching}\label{sec:application:MWM}

In \cite{bernstein2021framework}, Bernstein et al.~provided a framework for reducing dynamic \emph{weighted} matching to dynamic unweighted matching. 
For bipartite graphs, this reduction takes an $\alpha$-approximate unweighted matching algorithm with update time $T$, and constructs from it an $\alpha(1+\eps)$-approximate weighted matching algorithm with update time roughly $\tilde{O}_\eps(T)$. 
For general graphs, their result in general loses a $3/2$ factor in the approximation ratio, due to the integrality gap of the fractional matching polytope in bipartite graphs. 
Applying this directly to \Cref{thm:unweighted-kernel} yields a $(3+\eps)$-approximate weighted matching algorithm with update time $\tilde{O}_\eps(\sqrt[4]{m})$.

Bernstein et al.~provide improved bounds for their reduction when applied to kernel-based algorithms, for which they prove that the factor of $3/2$ loss need not be incurred. 
In particular, given a kernel maintenance algorithm $\calA$, their work yields $(2+\eps)$-approximate \emph{weighted} matching algorithms with roughly the same running time as $\calA$. More precisely, letting $\gamma_\eps :=(1/\eps)^{1/\eps}$, their work implies the following.

\begin{restatable}{lem}{kernelreduction}(Implicit in \cite{bernstein2021framework})\label{kernel-reduction}
    Let $\calA$ be a algorithm for maintaining an $(\eps,d)$-kernel $K$ for $\eps\in (0,1/12)$ and  $d \geq 288 \eps^{-3} \log(1/\eps)$, with update time $T(\eps,d,n,m)$ and $C(\eps,d,n,m)$ changes to $K$ per update.
    Then, there exists a $(2+O(\eps))$-approximate dynamic weighted matching algorithm $\calA'$ with update time $$O_\eps(T(\eps,d,\gamma_\eps n, \gamma_\eps m)+d\cdot C(\eps,d,\gamma_\eps n, \gamma_\eps m))\cdot \log (R).$$ 
    If $\calA$ is deterministic/adaptive, then so is $\calA'$. Moreover, if the bounds on the update time of $\calA$ and number of changes to $K$ hold in the worst-case, then so does the bound on the running time of $\calA'.$
\end{restatable}

This lemma immediately implies our dynamic weighted matching algorithm of \Cref{thm:weighted-kernel}, thus completing our proof of \Cref{thm:fast-kernel-algo}.

\begin{proof}[Proof of \Cref{thm:weighted-kernel}]
    Let $d=\sqrt[4]{m}\geq 288\cdot \eps^{-3}\log(1/\eps)$.\footnote{We focus on such sufficiently large $m$, since otherwise we can trivially maintain a $(1+\eps)$-approximate weighted matching in time near-linear in $m\leq f(\eps)$ by using the static algorithm of \cite{duan2014linear} after each update.} By \Cref{kernel-maintenance}, we can deterministically maintain an $(\eps,d)$-kernel with update time $T(\eps,d,n,m)=O(\sqrt{m}/(d\cdot \poly(\eps))=O(\sqrt[4]{m}/\poly(\eps))$ and $C(\eps,d,n,m)=O(1)$ changes per update to $G$, with both bounds holding in the worst case. Plugging these bounds into \Cref{kernel-reduction}, the theorem follows.
\end{proof}

\begin{remark}
    As noted in \cite{bernstein2021framework}, the dependence on $R$ is only logarithmic in $n$, provided the weights are polynomially bounded. 
    If this is not the case, the dependence on $R$ can be decreased exponentially to a doubly-logarithmic $\log\log R$ term, at the cost of extra logarithmic factors in $n$, using the reduction of  \cite{stubbs2017metatheorems}. We refer the reader to \cite[Appendix A]{bernstein2021framework} for a discussion.
\end{remark}

Since the reduction of \Cref{kernel-reduction} is not stated in a black-box manner in \cite{bernstein2021framework}, we outline its proof below, re-stating verbatim or slightly re-wording several of the definitions and lemmata of \cite{bernstein2021framework}.

\subsection{Folding and Unfolding}
In \cite{kao2001decomposition}, Kao et al.~provided a black-box reduction from (integer-)weighted bipartite matching to its unweighted counterpart, with a blow-up of $O(W)$ in the running time, for $W:=\max_e w(e)$.
(This result has since been generalized to general graphs, using other ideas, by \citet{pettie2012simple,huang2012efficient}.) 
The key idea of Kao et al., which Bernstein et al.~make use of, is the notion of graph folding and re-folding, defined below.

\begin{Def}
Let $G = (V, E, w)$ be a graph with non-negative
integral edge weights. The \emph{unfolded graph} $\phi(G)$ is an unweighted graph, with vertex set 
$V(\phi(G)):=\left\{u^i \mid u\in V, i\in [\max_{e\ni u} w(e)]\right\}$ and edge set 
$E(\phi(G)):=\left\{(u^i,v^{w(e)-i+1})
 \mid (u,v)\in E,\, i\in [w(e)]\right\}$ in $\phi(G)$.
\end{Def}
The key result of \citet{kao2001decomposition} is that in bipartite graphs $G$, the maximum (unweighted) matching in $\phi(G)$ has size equal to the maximum weight of a matching in $G$.
In general graphs, \citet{bernstein2021framework} prove similar relationships, by considering the following ``reverse'' operation of folding.
\begin{Def}
    Let $G=(V,E,w)$ be a weighted graph and let $H$ be a subgraph of $\phi(G)$. Then the \emph{refolded graph} $R(H)$ has vertex set $V$ and edge set $E(R(H)):=\{(u,v)\in E\mid  (u^i,v^{w(u,v)-i})\in E(H)\}$.
\end{Def}

\citet{bernstein2021framework}'s framework is particularly useful given a subgraph $H$ of the folded graph whose re-folded counter-part contains a large weighted matching, as in the following definition.
\begin{Def}
    A subgraph $H$ of $\phi(G)$ is \emph{$\alpha$-refolding approximate} if 
    $$MWM(G)\leq \alpha\cdot MWM(R(H)).$$
\end{Def}

As \citet{bernstein2021framework} show, kernels are precisely such refolding approximate subgraphs.
\begin{lem}[\cite{bernstein2021framework}] \label{refolding:kernels}
Let $G$ be a graph with edge weights in $[W]$ and let $d \geq 144 \eps^{-2} \log(2W/\eps)$, and $\eps\leq 1/12$. If $H$ is an $(\eps,d)$-kernel of $\phi(G)$, then $H$ is $1/(1/2 - 3\eps/2)\leq (2+4\eps)$-refolding-approximate
\end{lem}

The usefulness of refolding approximate subgraphs for dynamic matching algorithms is made apparent by \cite[Theorem 4.4]{bernstein2021framework}.\footnote{The statement of \cite[Theorem 4.4]{bernstein2021framework} available at \url{https://zlangley.com/a-framework-for-dynamic-matching-in-weighted-graphs.pdf} contains a typo, requiring a bound on \emph{recourse}. However, recourse bounds are not defined for subgraph maintenance algorithms, nor necessary in their proof.
We thank Aaron Bernstein for confirming this point (private communication).}

\begin{lem}[\cite{bernstein2021framework}]\label{BDL:refolding-reduction}
    Let $G$ be an $n$-node and $m$-edge graph
    with weights in $[W]$, and let $\eps > 0$. If there is an algorithm $\calA_u$ that maintains an $\alpha$-refolding-approximate degree-$\Delta$ subgraph $H$ of $\phi(G)$ over updates to $G$ with update time $T_u(n,m, \alpha,W)$ (per update to $G$) and with  $C_u(n,m,\alpha,W)$ updates to $H$ per update in $G$, then there is an algorithm $\calA_w$ that maintains a
    $(1 - \eps)\alpha$-approximate maximum weight matching in $G$ with update time $$O_\eps((T_u(\gamma_\eps n,\gamma_\eps m, \alpha,\gamma_\eps) +  \Delta\cdot C_u(\gamma_\eps n,\gamma_\eps m, \alpha,\gamma_\eps )) \cdot \log(R)).$$
    Furthermore, if $\calA_u$ is deterministic, so is $\calA_w$, and if the update time and number of changes to $H$ of $\calA_u$ are worst-case, so is the update time of $\calA_w$.
\end{lem}

\Cref{kernel-reduction} follows directly by combining lemmas \ref{refolding:kernels} and \ref{BDL:refolding-reduction}, and re-parameterizing $\eps$, while noting that $\log(2\gamma_\eps/\eps)\leq 2\cdot (1/\eps)\log(1/\eps)$ for $\eps\leq 1/12$.

\section{Bounding Maximum Degree by Arboricity}\label{sec:arboricity}
In this section, we show how we can replace $\Delta$ in our running time by $\sqrt{m}$ by using the algorithm by Solomon \cite{solomon2018local}. 

\begin{lem}\label{lem:arboricity}
Let $G$ be a dynamic graph and $\Delta$ be its maximum degree. If there exists an $\alpha$-approximate dynamic matching algorithm $\calA$ with worst-case update time of $T(\Delta, n, m)$, then there exists a $(1+\eps)\alpha$-approximate dynamic matching algorithm $\calA'$ with worst-case update time of $T(\sqrt{m}/\eps, n, m) + \log n$. If Algorithm $\calA$ is deterministic (randomized against an adaptive adversary), then so is $\calA'$.
\end{lem}

\begin{proof}
Suppose that each node of the graph selects $\sqrt{m}/ \eps$ of its edges (if it is possible, otherwise it selects all its neighbors) and we create a graph $G'$ by choosing edges that are selected by both of its endpoints. Maximum degree in $G'$ is at most $O(\sqrt{m} / \eps)$ since each node selects at most $\sqrt{m} / \eps$ edges. In Theorem 3 of \cite{solomon2018local}, Solomon shows that $\mu(G) \leq (1+\eps)\cdot \mu(G')$. 


It remains to show how to maintain $G'$ deterministically in a dynamic setting with logarithmic worst-case update time. Moreover, we should maintain $G'$ in such a way that each update to $G$ causes $O(1)$ updates to $G'$. For each node, we store its selected and non-selected edges in two different self-balancing BSTs. After an insertion into the graph $G$, for each of its endpoints, we insert the edge to the corresponding BST based on the degree of the node. If both endpoints add the edge to the selected-edges BST, we add this edge to the $G'$. After a deletion from the graph, for each of its endpoints, if the edge is in the selected-edges BST, we remove the edge from the selected edges and add another edge from the non-selected edges (if it is possible). If the newly added edge is selected by the other endpoint, we add it to the $G'$. Therefore, after each update to the input graph $G$, we have $O(1)$ updates in $G'$ with $O(\log n)$ worst-case update time.
\end{proof}

\section{Reduction from General to Bipartite Maximum Matching}\label{sec:bipartite:reduction}

In this section, we show a simple overlooked reduction from dynamic matching in general graphs to dynamic bipartite matching with a loss of a $\frac{3}{2}$ factor in the approximation---equal to the integrality gap of the fractional matching polytope. The reduction relies on the classic doubling of Nemhauser and Trotter, useful to prove $\frac{1}{2}$-integrality (and hence $\frac{3}{2}$-integrality gap) of the fractional matching polytope \cite{Nemhauser1975}.

\begin{lem}
Let $\mathcal{A}$ be an $\alpha$-approximation dynamic bipartite matching algorithm with $O(T)$ worst-case update time. Then there exists an $\frac{3\alpha}{2}$-approximation dynamic matching algorithm for general graphs with $O(T)$ worst-case update time.
\end{lem}

\begin{proof}
Let $G = (V, E)$ be a dynamic general graph. We define a bipartite graph $G'=(V^1 \cup V^2, E')$, where $V = \{v^i \lvert v \in V \}$ and $E' = \{ 
(u^1, v^2), (u^2, v^1) \lvert (u, v) \in E\}$. Each update to $G$ causes 2 updates in $G'$. Moreover, clearly $\mu(G') \geq 2\mu(G)$. Therefore, we can maintain an $\alpha$-approximation matching $M'$ in $G'$ with $O(T)$ update time which has size at least $|M'| \geq \frac{1}{\alpha}\cdot\mu(G') \geq \frac{2}{\alpha}\cdot\mu(G)$. Now consider a fractional matching $X_{uv} = \frac{\mathds{1}[(u^1, v^2)\in M'] + \mathds{1}[(u^2, v^1)\in M']}{2}$. Clearly, $\sum_{e}X_e = \frac{1}{2}\cdot|M'| \geq \frac{1}{\alpha}\cdot\mu(G)$. On the other hand, by the $\frac{3}{2}$-integrality gap of the fractional matching polytope \cite{Nemhauser1975}, we have that
$$
\mu(\text{support}(X)) \geq \frac{2}{3}\cdot \sum_e X_e \geq \frac{2}{3\alpha} \cdot \mu(G),
$$
where $\text{support}(X)$ is the graph with set of edges that their values are equal to $\frac{1}{2}$ or 1. Hence, a maximum matching in $\text{support}(X)$ (which changes by $O(T)$ edges per update in $G$) is a $\frac{3\alpha}{2}$-approximation matching in $G$. Since the maximum degree of $\text{support}(X)$ is at most 2, it consists of disjoint paths and cycles. Therefore, it is trivial to maintain maximum matching in $\text{support}(X)$ in $O(1)$ per update in $\text{support}(X)$, by maintaining two different lists of even/odd edges of the paths and cycles (for odd length cycles, we store one edge separately). A deletion into $\text{support}(X)$ causes a path to break into two different paths and a cycle to convert to a path. An insertion into $\text{support}(X)$ causes a path to become a larger path or create a new one (we cannot have an insertion incident to a cycle). All the previous operations can be done in $O(1)$. Therefore, we can maintain maximum matching of $\text{support}(X)$ in $O(T)$ per update in $G$.
\end{proof}

As a corollary of the above lemma, together with the bipartite algorithm of \cite{bernstein2015fully}, we immediately obtain a deterministic $(\frac{9}{4}+\eps)$-approximate dynamic matching algorithm in general graphs with update time $O(\sqrt[4]{m}\cdot \poly(1/\eps))$.

\bibliographystyle{plainnat}
\bibliography{main_arxiv} 

\end{document}